\documentclass[11pt]{article}

\usepackage{amsthm,amsmath,amssymb}
\usepackage{xcolor}
\usepackage{graphicx}
\usepackage{booktabs}
\usepackage{fullpage}
\usepackage{caption}
\usepackage{hyperref}
\usepackage[shortlabels]{enumitem}
\usepackage{lineno}

\renewcommand{\phi}{\varphi}

\def\operation#1{\textsf{#1}}

\def\original#1{\overline{#1}}

\newtheorem{lemma}{Lemma}

\newtheorem{problem}{Problem}
\newtheorem{theorem}{Theorem}
\newtheorem{definition}{Definition}
\newtheorem{corollary}{Corollary}
\newtheorem{observation}{Observation}

\begin{document}

\title{Atomic Embeddability, Clustered Planarity, and Thickenability\thanks{Research supported in part
by Austrian Science Fund (FWF): M2281-N35, and the National Science Foundation awards CCF-1422311, CCF-1423615, and DMS-1800734..}
}

\author{Radoslav Fulek\thanks{University of Arizona, Tucson, AZ, USA.}
\and
Csaba D. T\'oth\footnotemark[2]~\thanks{Department of Mathematics, California State University Northridge, Los Angeles, CA, USA and
Department of Computer Science, Tufts University, Medford, MA, USA.}
}

\maketitle
\thispagestyle{empty}

\begin{abstract}
We study the atomic embeddability testing problem, which is
a common generalization of clustered planarity (c-planarity, for short) and thickenability testing,
and present a polynomial-time algorithm for this problem, thereby giving the first polynomial-time
algorithm for c-planarity.

C-planarity was introduced in 1995 by
Feng, Cohen, and Eades as a variant of graph planarity, in which the
vertex set of the input graph is endowed with a hierarchical clustering
and we seek an embedding (crossing free drawing) of the graph in the plane
that respects the clustering in a certain natural sense. Until now, it
has been an open problem whether c-planarity can be tested efficiently,
despite relentless efforts.
The thickenability problem for simplicial complexes emerged in the topology
of manifolds in the 1960s. A 2-dimensional simplicial complex is thickenable
if it embeds in some orientable 3-dimensional manifold.  Recently, Carmesin
announced that thickenability can be tested in polynomial time.

Our algorithm for atomic embeddability combines ideas from Carmesin's work
with algorithmic tools previously developed for weak embeddability testing.
We express our results purely in terms of graphs on surfaces, and rely on the
machinery of topological graph theory.

Finally, we give a polynomial-time reduction from atomic embeddability to thickenability thereby showing that both problems are polynomially equivalent, and
show that a slight generalization of atomic embeddability to the setting in which
clusters are toroidal graphs is NP-complete.
\end{abstract}

\setcounter{page}{1}

\section{Introduction}
\label{sec:intro}

\paragraph{Clustered planarity} (for short, \textbf{c-planarity}) was introduced in 1995 by Feng, Cohen, and Eades~\cite{FCE95a_how,FCE95b_planarity}, motivated by applications in set visualization. Lengauer~\cite{L89_cplanarity} considered one of its variants already in the 1980s. The problem can be seen as a
hierarchical variant of planarity testing; a problem for which a linear-time algorithm has been known for a long time~\cite{HoTa74_planarity}.
In the extensive literature devoted to c-planarity and its variants, the complexity status of only restricted special cases has been established, most notably in~\cite{AFT19+_weak,AngeliniL19,CdBFPP08_CplanConnected,GUT02_2cl},
see also the somewhat outdated survey~\cite{CB05_invited}.
The c-planarity problem is formally stated as follows.
\begin{problem}[\textsc{C-planarity}]\label{prob:cplanarity}
We are given a simple graph $G=(V,E)$; a collection $\mathcal{C}$ of pairwise disjoint simple closed curves in $\mathbb{R}^2$; and a map $\mu: V \rightarrow \mathcal{R}$, where $\mathcal{R}$ is the set of path connected components (called \textbf{regions}) of $\mathbb{R}^2\setminus \bigcup_{C\in \mathcal{C}}C $.
Decide whether there exists an embedding $\psi$ of $G$ in $\mathbb{R}^2$ such that $\psi(v)\in \mu(v)$ and $|\psi(e)\cap C|\le 1$ for every $C\in \mathcal{C}$ and every edge $e\in E$.
\end{problem}

\paragraph{Embeddability in $\mathbb{R}^3$ and thickenability.}
Note that a graph is a 1-dimensional simplicial complex. We consider the 2-dimensional analog of graph embeddings in $\mathbb{R}^3$. It is a well-known result that for every graph $G$ there exists an orientable surface (an orientable 2-dimensional manifold) $S$ such that $G$ embeds in $S$. An analogous result fails for 2-dimensional polyhedra (2-polyhedra for short) and 3-dimensional manifolds (3-manifolds for short).
A 2-polyhedron $P$ is \textbf{thickenable} if $P$ embeds\footnote{In this case, topological and piecewise linear embeddability are equivalent.} in \emph{some} orientable 3-manifold.
It was known at least since the 1960s that testing whether a 2-polyhedron is thickenable is in NP, which is an immediate consequence of a theorem by Neuwirth~\cite{N68_thick}; see also~\cite{Skop94_thick}.
We restate it as Theorem~\ref{thm:Neuwirth} in Section~\ref{sec:c-planarity}
(in essence, it characterizes thickenability in terms of so-called \emph{link} graphs).
We remark that Neuwirth's theorem has been recently used in~\cite{MSTW18_3D} in the first step
of an algorithm that decides (not necessarily in polynomial time) whether a given 2-polyhedron embeds to $\mathbb{R}^3$.

The thickenability problem is formulated as follows. Let $H$ be a finite multigraph without loops (multiple edges are allowed). Let $P=(H,F)$ denote a 2-dimensional \textbf{(abstract) polyhedron},
where $F$ is a set of cycles in $H$. We assume that every edge of $H$ is contained in at least one element of $F$. The multigraph $H$ is the \textbf{1-skeleton} of $P$ and every element of $F$ is a \textbf{facet} in $P$. Note that $H$ might contain a cycle that is not a facet of $P$.

\begin{problem}[\textsc{Thickenability}]\label{prob:thickenability}
Given a 2-polyhedron $P=(H,F)$, where $H$ is a multigraph without loops and $F$ is the set of facets of $P$,
decide whether $P$ embeds to some orientable 3-manifold.
\end{problem}

Recently, Carmesin~\cite[Section 6]{C17_embed} announced that one can test whether a simply connected 2-polyhedron embeds in $\mathbb{R}^3$ in quadratic time\footnote{The correctness of the claimed running time has not been confirmed, yet.}, while embeddability testing for general 2-polyhedra in $\mathbb{R}^3$ is known to be NP-hard~\cite{dMRST18_embed3D}.
In the case of simply connected 2-polyhedra, thickenability is equivalent to embeddability in $\mathbb{R}^3$, see for example~\cite{C17_embed2}. Though this equivalence appears to be a well-known consequence of Perelman's famous result~\cite{P02_finite,P02_entropy,P02_ricci}, see also the book~\cite{M07_ricci}.

In fact, Carmesin's approach deals exclusively with testing the thickenability condition in Theorem~\ref{thm:Neuwirth} (cf.\ Section~\ref{sec:c-planarity}).
Therefore his approach to the embeddability in $\mathbb{R}^3$ applies also to \textsc{thickenability} in general, but it is
restricted simply connected 2-polyhedra. In contrast, testing whether a given 2-polyhedron, that is homeomorphic to a nonorientable surface, embeds in a given 3-manifold (that is, both the 2-polyhedron and the 3-manifold are given), is already NP-hard~\cite{BdMW17_nonOrient}.

\paragraph{Atomic embeddibility} was introduced in~\cite{AFT19+_weak} and~\cite[Section 11]{FK17+_ht}, see also~\cite{FK17_htsocg}, as a common generalization of \textsc{C-planarity} and \textsc{thickenability}. It is an extension of the concept of \emph{weak embeddability}~\cite{AFT19+_weak} (also known in topology as \emph{approximating simplicial maps by embeddings}~\cite{M94_derivative,ReSk98_deleted,S69_realization,Sko03_approximability}).
We do not define weak embedding here, but remark that its study in computational geometry was motivated
by the special case of a (piecewise linear) weak embeddings of a cycle in the plane,
which corresponds to \emph{weakly simple polygons}~\cite{AAET17,ChEX15,CdBPP09}.

Let $G$ and $H$ be finite multigraphs without loops. To distinguish between $G$ and $H$ in our terminology, the vertices and edges of $H$ are called \textbf{atoms} and \textbf{pipes}, respectively.
A map $\varphi:~G~\rightarrow~H$ is \textbf{simplicial} if it maps vertices to vertices (i.e., to atoms),
edges to vertices or edges (i.e., to atoms or pipes), and preserves edge-vertex incidences.
An \textbf{instance} of atomic embeddability is given by a simplicial map $\varphi: G \rightarrow H$.

The \textbf{thickening}
$\mathcal{H}$ of $H$ is an orientable 2-dimensional  surface constructed as follows. For each atom $\nu\in V(H)$, let $\mathcal{S}(\nu)$ be a 2-sphere with $\deg(\nu)$ pairwise disjoint open discs, called \textbf{holes}, removed. We fix an orientation on $\mathcal{S}(\nu)$, and define an arbitrary one-to-one correspondence between the holes of $\mathcal{S}(\nu)$ and the pipes incident to $\nu$. The thickening $\mathcal{H}$ is obtained by gluing the surfaces $\mathcal{S}(\nu)$, $\nu\in V(H)$, as follows; see Fig.~\ref{fig:local} (left) for an illustration. For every pipe $\rho\in E(H)$, $\rho=\nu\mu$, identify the pair of boundaries of the holes corresponding to $\rho$ by an orientation reversing homeomorphism. In particular, if $\nu\mu\not\in E(H)$, then $\mathcal{S}(\nu)$ and $\mathcal{S}(\mu)$ are disjoint\footnote{The surface $\mathcal{H}$ is reminiscent of the ball-and-stick or space-filling models in molecular chemistry.}.

An embedding $\mathcal{E}:G\rightarrow \mathcal{H}$ is an \textbf{atomic embedding} of $G$ with respect to $\varphi$ if every vertex $v\in V(G)$ is embedded in $\mathcal{S}(\varphi(v))$; and every edge $uv\in E(G)$ is embedded as a Jordan arc in $\mathcal{S}(\varphi(u))\cup\mathcal{S}(\varphi(v))$ as follows: If $\varphi(u)\neq\varphi(v)$ then the Jordan arc representing $uv$ intersects the hole corresponding to the pipe $\varphi(uv)$ in exactly one point, which is a \emph{proper crossing}, or in other words, a transversal intersection.

\begin{problem}[\textsc{Atomic embeddability}]\label{prob:atomic}
Given a pair of multigraphs without loops, $G$ and $H$, and a simplicial map  $\varphi: G \rightarrow H$,
decide whether an {atomic embedding} of $G$ with respect to $\varphi$ exists.
\end{problem}

We remark that an instance $(H,F)$ of \textsc{thickenability} corresponds to an instance $(G,H)$ of \textsc{atomic embeddability},
where $H$ is the same graph both instances, and
$G$ is a vertex disjoint union of cycles (disjoint copies of the cycles in $F$).

\paragraph{Results.}
In this paper, we present a polynomial-time algorithm for atomic embeddability, thereby giving the first polynomial-time algorithm for c-planarity.
Our approach combines ideas from Carmesin's work~\cite{C17_embed} with algorithmic tools previously developed for weak embeddability testing.
In particular, the elementary operation ``stretch'' (defined below) is based on a similar operation in~\cite{C17_embed}.
However, by formulating the problem in terms of graphs on surfaces, our results are more general and perhaps more accessible to the broader community.
A polynomial-time algorithm for c-planarity implies that some other constrained planarity problems that have previously been reduced to c-planarity are tractable, as well; see~\cite{AdBFPR12_simul} and~\cite[Figure 4]{AngeliniDL16}.

We also consider a further generalization of \textsc{atomic embeddability} in which the surfaces $\mathcal{S}(\nu)$, $\nu\in V(H)$,
may have higher genus (by attaching additional handles), and show that this problem is NP-complete even if each surface $\mathcal{S}(\nu)$ is based on a torus rather than a sphere.

In the last section, we give a short polynomial-time reduction of \textsc{atomic embeddability} to \textsc{thickenability},
which shows that both problems are polynomially equivalent.

\paragraph{Simultaneous embeddability of two graphs.}
Angelini and Da Lozzo~\cite{AngeliniDL16} proved that there exists a polynomial time reduction to  \textsc{c-planarity} from \textsc{connected sefe-2}, the problem of deciding simultaneous embeddability of two graphs in the case when the intersection of the two graphs is connected (see Section~\ref{sec:c-planarity} for a formal statement of the problem). Therefore our algorithm gives a polynomial time algorithm for this problem.
The general version of the problem, known as \textsc{sefe-2}, where the intersection of the two graphs may be disconnected, is notoriously difficult. Introduced by Brass et al.~\cite{BCD+07+_sefe}, it subsumes most of the studied planarity variants~\cite[Figure 2]{Sch13b_towards}, and has generated considerable research activity~\cite{AdBFPR12_simul,BKR18_simul,BR16_simul,HRL13_simul}; see also~\cite{BKR13_simul} for a survey. 
Schaefer~\cite[Theorem 6.17]{Sch13b_towards} realized that \textsc{c-planarity} is reducible in polynomial time to the problem \textsc{sefe-2}. Therefore it is an unfortunate state of affairs that its complexity status is still unknown. Although, Carmesin's and our results give a new hope that a resolution of \textsc{sefe-2} problem might be within reach.

Let us note that there also exists a natural and fairly straightforward polynomial-time reduction of \textsc{connected sefe-2} to \textsc{thickenability}, which was found independently by de Mesmay, Kalu\v{z}a, and Tancer~\cite{BMW15,dMKT19_per} and these authors. This suggests that \textsc{thickenability}, and hence, \textsc{atomic embeddability}, is not powerful enough to solve \textsc{sefe-2} in general without using significantly novel ideas.

\paragraph{Organization.}
Section~\ref{sec:atomic} presents a polynomial-time algorithm for atomic embeddability.
Section~\ref{sec:compatible} shows that a further generalization of the problem is NP-hard.
In Section~\ref{sec:c-planarity}, we give a direct polynomial-time reduction of
\textsc{atomic embeddability} to \textsc{thickenability}, which also establishes a polynomial
reduction of another problem, \textsc{connected sefe-2}, to \textsc{thickenability}.

\section{Atomic Embeddings}
\label{sec:atomic}

In this section we present a polynomial-time algorithm for \textsc{atomic embeddability}. After defining local graphs, which are crucial for the algorithm, we present a high-level overview in Section~\ref{ssec:highlevel}.
Section~\ref{ssec:prelim} introduces additional terminology. We reduce a given instance $\varphi$ to normal form (defined below) in Section~\ref{ssec:ds}; and introduce five \emph{elementary operations} on atomic instances in Section~\ref{ssec:operations} that are used in our main algorithm. We show how to solve two special cases in linear time in Sections~\ref{ssec:toric} and~\ref{ssec:subcubic}. Our main algorithm in Section~\ref{ssec:main} reduces all normal instances to these special cases. We finish with a running time analysis in Section~\ref{ssec:runtime}.

\begin{figure}
\centering
\includegraphics[scale=1]{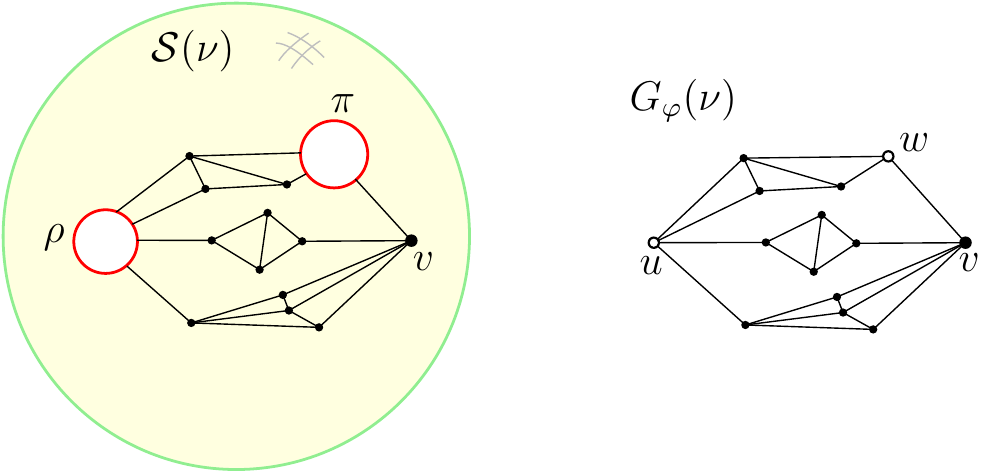}
\caption{Part of an atomic embedding of $G$ on $\mathcal{S}(\nu)$, where the atom $\nu$ is incident to pipes $\rho$ and $\pi$ (left), and the corresponding local graph $G_\varphi(\nu)$ (right). The virtual vertices $u$ and $w$ in $G_\varphi(\nu)$ correspond to the pipes $\rho$ and $\pi$, respectively. All other vertices in $G_\varphi(\nu)$ are ordinary.}
\label{fig:local}
\end{figure}

\paragraph{Local graphs.}
Let $\varphi:G\rightarrow H$ be an instance of {atomic embeddability}.
The simplicial map $\varphi:G\rightarrow H$ naturally extends to subgraphs of $G$. For an atom $\nu\in V(H)$, $\varphi^{-1}[\nu]$ denotes the subgraph of $G$ mapped to $\nu$ by $\varphi$. For a pipe $\rho\in E(H)$, $\varphi^{-1}[\rho]$ denotes the subset of edges of $G$ mapped to $\rho$ by $\varphi$.

For every atom $\nu\in V(H)$, we define a multigraph $G_\varphi(\nu)$, which captures the local structure of $\varphi$ at the atom $\nu$ and its incident pipes; see Fig.~\ref{fig:local} for an illustration. (We remark that graphs $G_\varphi(\nu)$ are analogous to the graphs $\overline{C}$ in~\cite{AFT19+_weak} and the links in~\cite{C17_embed2}.)

The vertices of $G_\varphi(\nu)$ are in a one-to-one correspondence with the union of the set of vertices in $V(G)$ mapped by $\varphi$ to $\nu$ (that is, $V(\varphi^{-1}[\nu])$) and the set of pipes incident to $\nu$. Hence, we can distinguish between \textbf{ordinary vertices} that correspond to vertices in $V(\varphi^{-1}[\nu])$ and \textbf{virtual vertices} that correspond to pipes incident to $\nu$.
For every edge in $E(G)$ between two vertices in $V(\varphi^{-1}[\nu])$ in $G$, add an edge in $G_\varphi(\nu)$ between the corresponding vertices. Finally, for every edge in $uv\in E(G)$ where $u\in V(\varphi^{-1}[\nu])$ and $v\not\in  V(\varphi^{-1}[\nu])$, add an edge  $G_\varphi(\nu)$ between the ordinary vertex $u$ and the virtual vertex corresponding to $\varphi(uv)$. Thus, edges of $G_\varphi(\nu)$ are in a one-to-one correspondence with the union of the edges of $G$ between vertices in $V(\varphi^{-1}[\nu])$ and the edges of $G$ mapped to pipes incident to $\nu$ by $\varphi$. Let $\original{e}\in E(G)$ denote the edge corresponding to an edge $e\in E(G_\varphi(\nu))$.

Note that the virtual vertices form an independent set in $G_\varphi(\nu)$.
An embedding $\mathcal{E}_{\nu}$ of $G_\varphi(\nu)$ is \textbf{inherited} from an atomic embedding $\mathcal{E}$ of $G$, if $\mathcal{E}_\nu$ is obtained from the restriction $\mathcal{E}$ to $\mathcal{S}(\nu)$ by filling the holes of $\mathcal{S}(\nu)$ with discs, and then contracting them  to points.

Let $\mathcal{E}:G\rightarrow \mathcal{S}$ be an embedding of a graph on an orientable surface.
The \textbf{rotation} at a vertex $v\in V(G)$ is the counterclockwise cyclic order of the end pieces of the edges incident to $v$. The \textbf{rotation system} of $\mathcal{E}$ is the set of rotations of all vertices of $G$.
A vertex $v$ of a planar graph has a \textbf{fixed rotation} (for short, is \textbf{fixed}) if its rotation in every embedding of the graph in the plane is unique up to the choice of orientation.
For a pair of virtual vertices $u$ and $v$ of $G_\varphi(\mu)$ and $G_\varphi(\nu)$, resp., corresponding to a pipe $\mu\nu=\rho\in E(H)$,
the edges $e$ incident to $u$ (resp., $v$) correspond to edges $\overline{e}\in \varphi^{-1}[\rho]$. The rotations of $u$ and $v$ in an embedding of all local graphs in the plane determine two cyclic orders $\varphi^{-1}[\rho]$. This allows us, in particular, to define that the rotations at $u$ and $v$ to be \textbf{opposite} (in other words, \emph{reverse}) to each other, if the rotation at $u$ is $(uu_1,\ldots, uu_{\deg(u)})$, at $v$ it is $(vv_{\deg(v)},\ldots, vv_{1})$, and $\overline{uu_i}=\overline{vv_i}=u_iv_i\in E(G)$. The rotations of $u$ and $v$ are \textbf{compatible} if they are the same or opposite to each other; and \textbf{incompatible} otherwise.

An instance $\varphi$ of \textsc{atomic embeddibility} is \textbf{positive} if there exists an atomic embedding of $G$ with respect to $\varphi$. Two instances, $\varphi$ and $\varphi'$, are \textbf{equivalent} if $\varphi$ and $\varphi'$ are both positive, or $\varphi$ and $\varphi'$ are both negative.
We can now formulate \textsc{atomic embeddibility} in terms of the rotation systems of plane embeddings of the graphs $G_\varphi(\nu)$, $\nu\in V(H)$.

\begin{observation}\label{obs:atomic}
An instance $\varphi:G\rightarrow H$ of atomic embeddability is positive if and only if the graphs $G_\varphi(\nu)$, $\nu\in V(H)$, are planar, and they each have embeddings in the plane such that for every pipe $\mu\nu\in E(H)$, the virtual vertices corresponding to $\mu\nu$ in $G_\varphi(\mu)$ and $G_\varphi(\nu)$ have opposite rotations (in the  sense that we consider every edge $e$ incident to a  virtual vertex as $\overline{e}$).
\end{observation}

 For a subset $V'\subset V(G)$, a \textbf{$V'$-bridge} $B$ in $G$ is a subgraph of $G$ obtained as the union of $V'$, a connected component $C$ of $G\setminus V'$, and all the edges joining a vertex of $C$ with a vertex of $V'$. We allow $B$ to consist of a single edge between two vertices in $V'$, or a loop incident to a vertex in $V'$ (see Fig.~\ref{fig:para-para}(left)).
Analogously to Carmesin~\cite{C17_embed2}, we also define two special graphs (as possible local graphs). A \textbf{p-path} is a graph that consists of two vertices (\textbf{poles}) connected by one or more subdivided edges (Fig.~\ref{fig:para-para}(middle)). A \textbf{p-star} is a graph with a unique cut vertex (\textbf{center}) whose bridges are p-paths with one pole at the center (Fig.~\ref{fig:para-para}(right)).

\subsection{High Level Overview of the Recognition Algorithm}
\label{ssec:highlevel}

Given an instance $\varphi$ of atomic embeddiblity, we apply a sequence of elementary operations that each produces an equivalent instance $\varphi'$ (with respect to atomic embeddability). Intermediate steps of our algorithm may detect that the instance is negative when a local graph $G_\varphi(\nu)$ is nonplanar. It may also disconnect the graph $H$, effectively splitting an instance into independent instances. Ultimately, it reduces $\varphi$ to a family of instances, each of which is either toroidal (where both $G$ and $H$ are 2-regular), or subcubic (where the maximum degree of all local graphs is at most 3). In both cases, we can easily test atomic embeddibility in linear time (Sections~\ref{ssec:toric} and \ref{ssec:subcubic}).
Hence, the witness of atomic non-embeddibility that is provided by our algorithm is either the non-planarity of a local graph in an  instance produced by a sequence of elementary operations, or negative subcubic or toroidal instance.

 Let $\mathcal{G}^*$ denote the disjoint union of all local graphs except those belonging to toroidal subinstances. Our algorithm incrementally reduces the maximum degree $\Delta=\Delta(\varphi)=\max_{v\in V(\mathcal{G}^*)} \deg(v)$.
The two key operations for dealing with a vertex $v\in V(\mathcal{G}^*)$ of degree $\Delta\ge 4$ are \operation{Stretch($v,.$)}, which splits $v$ into two vertices of smaller degree (illustrated in Fig.~\ref{fig:rigid}), and \operation{Contract$(.)$}, which contracts a pipe (illustrated in Fig.~\ref{fig:useful}). Operation \operation{Stretch($.$)} can be applied to a virtual or an ordinary vertex: If it is applied to an ordinary vertex, it modifies only $G$ and not $H$, but if it is applied to a virtual vertex, it modifies both $G$ and $H$, and in particular it increases the genus of the surface $\mathcal{H}$ by 1. We note that the increase in the genus of $\mathcal{H}$ occurs also in the special case that $\varphi$ represents an instance of c-planarity (when $\mathcal{H}$ is initially homeomorphic to a 2-sphere, i.e., when its genus is 0). This explains in part why this approach for the inherently planar problem of c-planarity has not been considered before. The generalization of c-planarity to surfaces of higher genus allows for a broader range of operations, but it also poses several technical challenges that had to be resolved---some of them even indicated that the problem might be NP-complete, which we discuss next.

Unfortunately, \operation{Stretch($v,.$)} produces an equivalent instance only if we already have some partial information about the rotation of vertex $v$. In general, it cannot reduce the degree of a cut vertex. This obstacle is overcome with the help of a surprisingly simple operation, \operation{Contract$(\rho)$}, which contracts a pair of atoms in $H$ joined by a single pipe $\rho$ into one atom, thereby eliminating a pair of virtual vertices in $\mathcal{G}^*$ corresponding to $\rho$.
An almost identical operation is also crucial in our recent joint work with Akitaya~\cite{AFT19+_weak} about weak embeddability.
Nevertheless, the possibility of using this operation in the context of (the general case of) c-planarity or atomic embeddability was not clear to us for some time. The reason is that the operation \operation{Contract$(\rho)$} for a pipe $\rho=\mu\nu$ can only be applied in a very restricted setting, essentially if and only if $G_\varphi(\mu)$ or $G_\varphi(\nu)$ is a p-path and $\rho$ corresponds to a pole of that p-path; or if they are both p-stars and $\rho$ corresponds to their centers. The crucial observation that saves the day, which is implicit in Carmesin's work, is that after some preprocessing that resolves 2-cuts with a vertex of degree $\Delta$, we can use the operation \operation{Enclose(.)}, illustrated in Fig.~\ref{fig:1cut}, to turn each cut vertex of degree $\Delta$ into a center of a p-star.

In order to show that our algorithm runs in polynomial time, we define a nonnegative potential $\Phi(\varphi)$ bounded from above throughout the execution of the algorithm by a polynomial function of $|V(G)|$ that strictly decreases after every application of \operation{Stretch($.$)} or \operation{Contract$(.)$}, but unfortunately, not after every application of \operation{Enclose(.)}, which possibly just creates a pair of new virtual vertices in $\mathcal{G}^*$. Hence, we had to design a charging scheme that controls the growth of $\mathcal{G}^*$.

Several other similar, but less crucial, operations are used in the preprocessing and postprocessing steps of the algorithm, where the preprocessing step normalizes the input instance in order to allow a relatively smooth runtime analysis, and the postprocessing step handles toroidal instances and subcubic instances (where $\Delta(\varphi)\leq 3$).

\subsection{Preliminaries}
\label{ssec:prelim}

Let $G$ and $H$ be multigraphs without loops (multiple edges are allowed in both $H$ and $G$).
By a slight abuse of notation, if there is no danger of confusion, we sometimes denote edges by unordered pairs of their endpoints (even though several edges may connect the same pair of vertices). A path, cycle, and walk in a graph is always assumed to be a sequence of edges (rather than vertices). Recall that in order to distinguish $G$ and $H$ in our terminology, the vertices and edges of $H$ are called {atoms} and {pipes}, respectively. We use the convention that vertices and edges of $G$ are denoted by lower case Roman letters (e.g., $u,v,z$ and $e,f,g$), respectively, and the atoms and pipes by lower case Greek letters (e.g., $\nu,\mu$ and $\rho,\pi$).

\paragraph{Cut vertices, 2-cuts, and 2-edge-cuts.}
 Every vertex of degree 2 or less has a unique rotation, hence it has a fixed rotation. For this reason, we use a topological notion of 1- and 2-cuts,  which is invariant to subdivisions of edges and supression of vertices of degree 2.
 For a connected graph $G$, which is not a cycle, denote by $G^-$ the multigraph obtained by supressing all vertices of degree 2.
Hence, $G^-$ is free of  \textbf{subdivided edges}, defined as paths whose internal vertices have degree 2.
Note that $G^-$ can have loops corresponding to cycles in $G$ that form leaf blocks.

Let $G$ be a connected graph that is not a cycle.
A vertex $v\in V(G^-) \subseteq V(G)$ is a \textbf{proper cut vertex} (or \textbf{proper 1-cut}) of $G$ if there are two or more $\{v\}$-bridges in $G^-$. A pair of vertices $\{u,v\}\subset V(G^-)\subseteq V(G)$ is a \textbf{proper 2-cut} of $G$ if there are at least three $\{u,v\}$-bridges in $G^-$, or there are exactly two $\{u,v\}$-bridges in $G^-$, neither of which is an edge in $G^-$. (Note that if there are exactly two $\{u,v\}$-bridges in $G^-$, and one of them is an edge between $u$ and $v$, then $\{u,v\}$ is not a 2-cut in $G^-$.) A pair of edges $\{e,f\}\subset E(G)$, such that at least one vertex incident to $e$ and one vertex incident to $f$ is of degree at least 3, is a \textbf{proper 2-edge-cut} of $G$ if  there exist edges $e^-$ and $f^-$, such that $e^-$ and $f^-$ were obtained by suppressing internal vertices of degree 2 of a path containing $e$ and $f$, respectively, and  $\{e^-,f^-\}$ is a 2-edge-cut in $G^-$. Finally, for a proper 2-cut $\{u,v\}$, a $\{u,v\}$-bridge $B$ is \textbf{separable} if $\deg_B(u)\in \{1,\deg(u)-1\}$ and $\deg_B(v)\in\{1,\deg(v)-1\}$, otherwise it is \textbf{nonseparable}.

\begin{observation}
\label{obs:proper2edgecut}
Let $G$ be a connected graph that is not a cycle.
If $\{u,v\}$ is a proper 2-cut and $B$ is a separable $\{u,v\}$-bridge but not a subdivided edge,
then there exists a proper 2-edge-cut $\{e,f\}$ in $G$ such that $u\in e$ and $v\in f$.
\end{observation}

We often tacitly use the following well-known result by Mac~Lane~\cite{M37_struc}. If $G$ is a connected planar graph, and the rotation of a vertex $v$ is not fixed, then $\deg(v)\geq 3$ and $v$ participates in a proper 1- or 2-cut.
In particular, if every graph $G_\varphi(\nu)$, $\nu\in V(H)$, is a subdivision of a 3-connected graph, we can use planarity testing to check the conditions in Observation~\ref{obs:atomic}, and easily reduce the atomic embedibility problem
to 2SAT (cf.~Section~\ref{ssec:subcubic}). The challenge is, therefore, to handle the possible rotations of vertices that participate in proper 1- or 2-cuts in some local graph $G_\varphi(\nu)$.

\begin{figure}
\centering
\includegraphics[width=\textwidth]{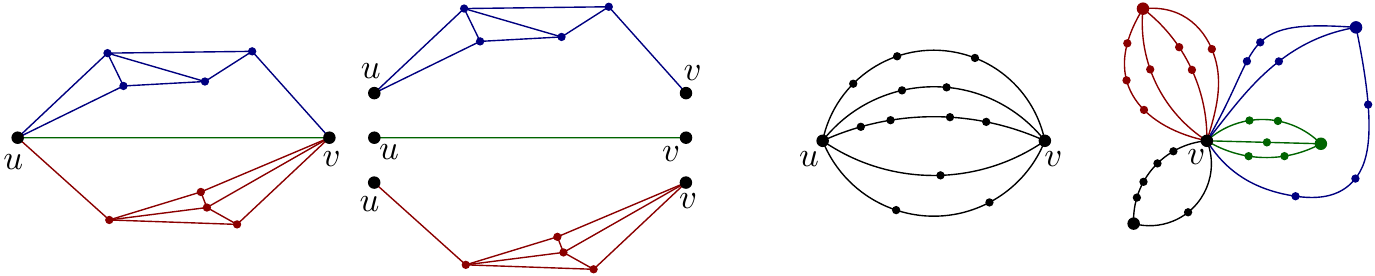}
\caption{A proper 2-cut $\{u,v\}$ and its three bridges (left), a p-path with poles $u$ and $v$ (middle), and a p-star centered at $v$ (right).}
\label{fig:para-para}
\end{figure}

\subsection{Preprocessing and Data Structures}
\label{ssec:ds}

Our algorithm uses a sequence of elementary operations that dynamically modify a given instance $\varphi:G\rightarrow H$ of atomic embeddability. For the running time analysis (Section~\ref{ssec:runtime}),
we need to maintain data structures that support these operations.
We assume that the input specifies $G$, $H$, and $\varphi$ explicitly (i.e., adjacency lists for the graphs $G$ and $H$, and pointers from the vertices and edges of $G$ to their images in $H$ under the map $\varphi:G\rightarrow H$).
The \textbf{size} of an instance $\varphi:G\rightarrow H$ is the total number of edges and vertices in the graphs $G$ and $H$.
Before we present our data structures (which do not maintain $H$ and $\varphi$ explicitly), we preprocess the instance $\varphi$.

\begin{definition}
An instance $\varphi:G\rightarrow H$ of atomic embeddability is \textbf{normal} if
\begin{itemize}\itemsep -1pt
\item the degree of every virtual vertex in every $G_\varphi(\nu)$, $\nu\in V(H)$, is 3 or higher; and
\item $G_\varphi(\nu)$ is connected for all $\nu\in V(H)$.
\end{itemize}
\end{definition}

\begin{figure}
\centering
\includegraphics[scale=1]{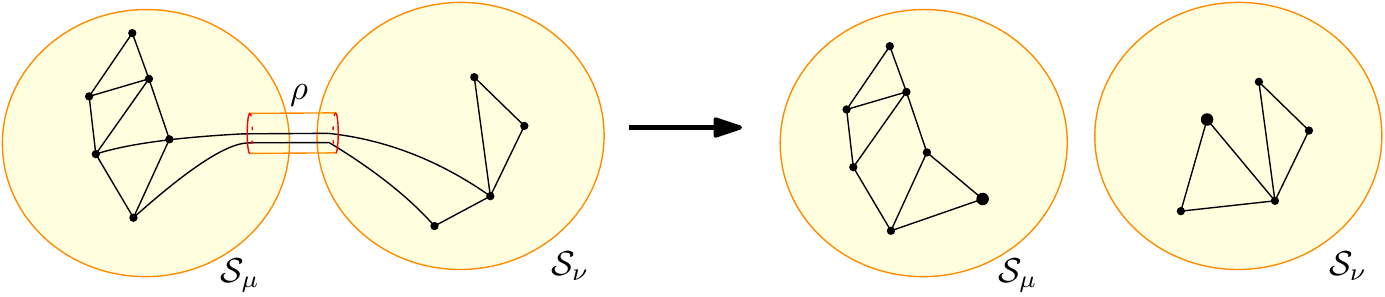}
\caption{An atomic embedding of $G$ on $\mathcal{S}(\mu)\cup\mathcal{S}(\nu)$ before and after operation \operation{Suppress($\rho$)}, where $\rho=\mu\nu$. The operation eliminates a pipe with at most two edges of $G$.}
\label{fig:suppress}
\end{figure}

We define an operation that eliminates pipes with 2 or less edges,
see Fig.~\ref{fig:suppress}.

\paragraph{\operation{Suppress$(\rho)$}.}
We are given a pipe $\rho\in E(H)$ such that $|\varphi^{-1}[\rho]|\leq 2$.
Let  $\mu,\nu\in V(H)$ be the two atoms incident to $\rho$.
Remove the pipe $\rho$ from $E(H)$.
If $\varphi^{-1}[\rho]$ contains one edge, say $uv\in E(G)$ with $\varphi(u)=\mu$ and $\varphi(v)=\nu$, then delete $uv$ from $E(G)$, insert two new vertices $u',v'$ and new edges $uu',vv'$ into $G$, and update $\varphi$ with $\varphi(u')=\mu$ and $\varphi(v')=\nu$.
If $\varphi^{-1}[\rho]$ contains two edges, say $u_iv_i\in E(G)$ with $\varphi(u_i)=\mu$ and $\varphi(v_i)=\nu$, for $i\in \{1,2\}$, then delete both $u_1v_1$ and $u_2v_2$ from $E(G)$, insert two new vertices $u',v'$ and new edges $u_1u'$, $u_2u'$, $v_1v'$, and $v_2v'$ into $G$, and update $\varphi$ with $\varphi(u')=\mu$ and $\varphi(v')=\nu$.

Since the virtual vertices that correspond to $\rho$ in $G_\varphi(\mu)$ and $G_\varphi(\nu)$ have fixed rotations, by Observation~\ref{obs:atomic}, the following is straightforward.

\begin{lemma}
\label{lem:supress-virtual}
For every instance $\varphi:G\rightarrow H$ of atomic embeddability, and every pipe $\rho\in E(H)$, whose corresponding vertices in local graphs have degree less than 3, operation \operation{Suppress($\rho$)} produces an equivalent instance.
\end{lemma}

\begin{figure}
\centering
\includegraphics[scale=1]{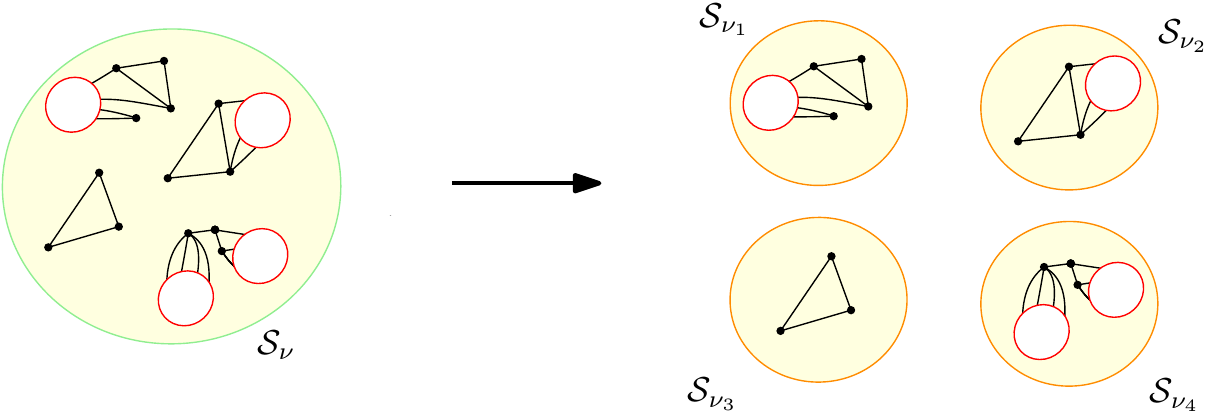}
\caption{An atomic embedding of $G$ on $\mathcal{S}(\nu)$ and $\bigcup_i\mathcal{S}(\nu_i)$ before and after, respectively, operation \operation{Split($\nu$)}. The operation splits an atom $\nu\in V(H)$, for which $G_\varphi(\nu)$ is disconnected, into as many atoms as the number of connected components in $G_\varphi(\nu)$.}
\label{fig:split}
\end{figure}

We define an operation that splits an atom $\nu$ if $G_\varphi(\nu)$ is disconnected, see Fig.~\ref{fig:split}.

\paragraph{\operation{Split$(\nu)$}.}
We are given a local graph $G_\varphi(\nu)$ whose connected components are $C_1,\ldots, C_{k}$, for some $k\in \mathbb{N}$. Delete $\nu$ from $H$, introduce new vertices $\nu_1,\ldots, \nu_{k}$ in $V(H)$, and introduce a pipe $\nu_i\mu$ for every $\rho=\nu\mu\in E(H)$ such that $\rho$ corresponds to a virtual vertex of $C_i$. Finally, redefine $\varphi$ on $V(\varphi^{-1}[\nu])$ as follows: Put $\varphi(v)=\nu_i$ if $v\in V(C_i)$.

By Observation~\ref{obs:atomic}, the following is straightforward.

\begin{lemma}
\label{lem:split-atom}
For every instance $\varphi:G\rightarrow H$ of atomic embeddability,
an application of \operation{Split($\nu$)} produces an equivalent instance.
\end{lemma}

\begin{enumerate}[(1)]
\item[] \textbf{Preprocessing$(\varphi)$.} Input: an instance $\varphi:G\rightarrow H$ of atomic embeddability.
\item\label{it:SubI1}
For every pipe $\rho\in E(H)$ with $|\varphi^{-1}[\rho]|\leq 2$, apply \operation{Suppress($\rho$)}.
\item\label{it:SubI2}
For every atom $\nu\in V(H)$, where $G_\varphi(\nu)$ is disconnected, apply \operation{Split($\nu$)}.
\end{enumerate}

\begin{lemma}\label{lem:Pre}
For an instance $\varphi:G\rightarrow H$ of atomic embeddability of size $n$,
Preprocessing runs in $O(n)$ time and returns an equivalent normal instance $\varphi'$.
\end{lemma}
\begin{proof}
By Lemmas~\ref{lem:supress-virtual} and~\ref{lem:split-atom}, the instance $\varphi'$ is equivalent to $\varphi$.
Step~\ref{it:SubI1} eliminates virtual vertices of degree less than 3, and Step~\ref{it:SubI2} does not change the degree of any vertex in local graphs.
Step~\ref{it:SubI2} splits the local graphs $G_\varphi(\nu)$, $\nu\in V(H)$, into connected components. Hence, $\varphi'$ is normal.
Step~\ref{it:SubI1} runs in $O(1)$ time for each pipe of degree less than 3.
Step~\ref{it:SubI2} runs in $O(m)$ time for every local graph $G_\varphi(\nu)$ with $m=m(\nu)$ edges; which yields an overall running time of $O(n)$.
\end{proof}

\paragraph{Data Structures.}
For a normal instance $\varphi:G\rightarrow H$, let $\mathcal{G}$ be the disjoint union of all local graphs $G_\varphi(\nu)$, $\nu\in V(H)$. We maintain the graphs $G$ and $\mathcal{G}$ by adjacency lists.
We maintain the set $V(H)$ of atoms implicitly: Each connected component in $\mathcal{G}$ corresponds to an atom $\nu\in V(H)$. We maintain the set $E(H)$ of pipes as follows: For every pipe $\rho\in E(H)$, we maintain two pointers to the two virtual vertices in $\mathcal{G}$ that correspond to $\rho$; and also maintain the set $\varphi^{-1}[\rho]\subset E(G)$ of edges mapped to $\rho$ in a doubly linked list. Furthermore, for each edge $uv\in \varphi^{-1}[\rho]$, with $\varphi(u)=\mu$ and $\varphi(v)=\nu$, we maintain a pointer to $\rho$, and to the edge in $G_\varphi(\mu)$ (resp., $G_\varphi(\nu)$) that joins the virtual vertex corresponding to $\rho$ and $u$ (resp., $v$).

For every connected component $G_\varphi(\nu)$ of $\mathcal{G}$, we maintain $G^-_\varphi(\nu)$ (i.e., the multigraph obtained by supressing vertices of degree 2), if $G_\varphi(\nu)$ is not a cycle, by adjacency lists. Furthermore, we maintain the block tree of $G^-_\varphi(\nu)$, which is a bipartite graph that represents incidences between cut vertices and blocks (i.e., maximal 2-connected components). For each block of $G^-_\varphi(\nu)$, we also maintain an SPQR decomposition tree introduced by Di Battista and Tamassia~\cite{DT89_spqr}, which is a hierarchical decomposition used for representing all 2-cuts and their bridges. For each vertex $v$ of $\mathcal{G}$, we maintain indicator variables that record whether $v$ is an ordinary or virtual vertex, whether it is a proper cut vertex or contained in a proper 2-cut.
At initialization, all these data structures can be computed in linear time in the size of $G$ and $H$.
The data structures can be updated in linear time if necessary. (Currently available dynamic data structures for planarity testing and SPQR-trees, with sublinear update times, support some but not all of our graph operations.)

As we shall see, whenever our algorithm creates a pipe of degree less than 3, it is immediately suppressed. If our algorithm modifies a graph $G_\varphi(\nu)$ in a way that it disconnects into components, then we assume that it immediately splits the corresponding atom $\nu$ as described above. In particular, our data structure supports the operation \operation{Split($\nu$)} in 0 time. In the remainder of the algorithm, we may assume that every instance of atomic embeddability is normal.

\subsection{Elementary Operations}
\label{ssec:operations}

In this section we describe operations used in our algorithm for a given instance $\phi:G\rightarrow H$ of atomic embeddibility. Each operation modifies the instance $\varphi$.
Each operation is local in the sense that it affects an atom $\nu$ and possibly one or two of its neighbors.
That is, the modifications incur changes in $G_\varphi(\nu)$, and possibly in $G_\varphi(\nu')$, for some of the neighbors $\nu'$ of $\nu$.

\begin{figure}[htbp]
\centering
\includegraphics[scale=1]{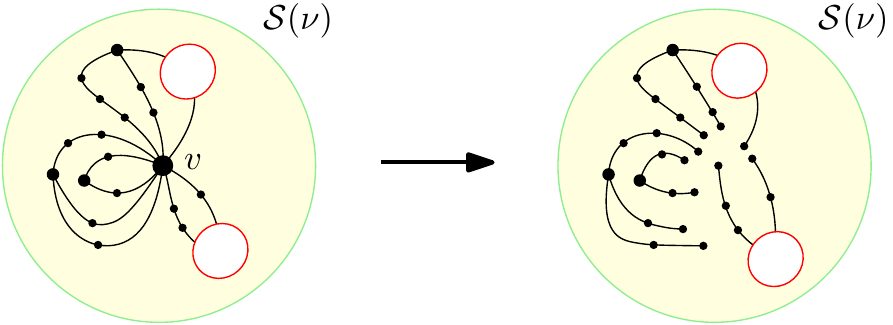}
\caption{An atomic embedding of $G$ on $\mathcal{S}(\nu)$ before and after operation \operation{Detach($v$)}. The operation turns an ordinary vertex $v$, the center of  p-star $G_\varphi(\nu)$, into $\deg(v)$ leaves.}
\label{fig:detach}
\end{figure}

The following operation turns an ordinary vertex $v$ into $\deg(v)$ leaves, see Fig.~\ref{fig:detach}.

\paragraph{\operation{Detach$(v)$}.}
Let $v$ be an ordinary vertex in a graph $G_\varphi(\nu)$ such that every $\{v\}$-bridge is a p-path
(that is, either $G_\varphi(\nu)$ is a p-star with center $v$, or $G_\varphi(\nu)$ is a p-path with a pole at $v$). Let $vu_1,\ldots, vu_{\deg(v)}$ denote the edges incident to $v$ in $G$. Remove $v$ and its incident edges from $G$. Then introduce $\deg(v)$ new vertices $v_1,\ldots, v_{\deg(v)}$ and add edges $u_iv_i$, for all $i\in [\deg(v)]$ to $G$. Finally, define $\varphi(v_iu_i)=\varphi(vu_i)$.

By Observation~\ref{obs:atomic}, the following is straightforward.

\begin{lemma}
\label{lem:split}
For an instance $\varphi:G\rightarrow H$ of atomic embeddability,
\operation{Detach($v$)} produces an equivalent instance $\varphi'$.
The operation can be implemented in $O(\deg(v))$ time.
\end{lemma}

In the following we define the operation of enclosing a bridge in $G_\varphi(\nu)$, see Fig.~\ref{fig:1cut}.
This operation is analogous to \emph{stretching of a local branch} in~\cite{C17_embed} except that we apply it in a more general setting.

\begin{figure}[htbp]
\centering
\includegraphics[scale=0.8]{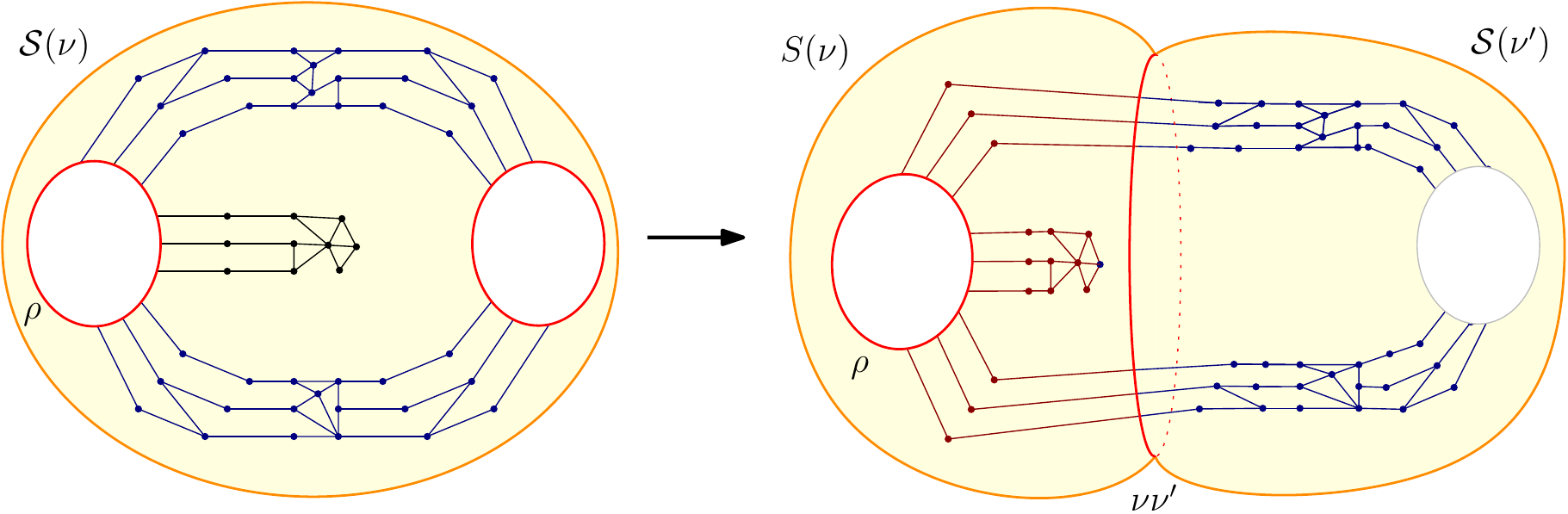}
\caption{An atomic embedding of $G$ on $\mathcal{S}(\nu)$  before and after operation \operation{Enclose($B$)}, where $B$ (colored blue) is a $\{v\}$-bridge of $G_\varphi(\nu)$ and $v$ is a virtual vertex corresponding to the pipe $\rho$.}
\label{fig:1cut}
\end{figure}

\paragraph{\operation{Enclose$(B)$}.}
We are given a $\{v_1,\ldots, v_k\}$-bridge $B$ in $G_\varphi(\nu)$.
The operation does not modify $G$ except for subdividing its edges.
We first describe the changes in $H$, and then the changes in the local graphs.
Create a new atom $\nu'$ and a new pipe $\nu\nu'$. Replace every pipe $\rho=\mu\nu$ that corresponds to a virtual vertex in $B\setminus \{v_1,\ldots, v_k\}$ with a new pipe $\mu\nu'$.  For every ordinary vertex $u\in V(B)\setminus \{v_1,\ldots, v_k\}$, set $\varphi(u)=\nu'$.
For every edge $e\in E(G)$, for which the pipe $\varphi(e)=\rho=\mu\nu$ has been replaced by $\rho'=\mu\nu'$, set $\varphi(e)=\rho'$.
If  $v_iu\in E(B)$, and $v_i$ or $u$ is a virtual vertex of $G_\varphi(\nu)$, then  subdivide $\original{v_iu}\in E(G)$ by a vertex $w$ and define $\varphi(w)$ as follows: If $v_i$ is virtual, then put $\varphi(w)=\nu$;
otherwise put $\varphi(w)=\nu'$.
Finally, update the definition of $\varphi$ on the edges of $B$ according to the value of $\varphi$ on the vertices of $G$
(this is uniquely determined since  $\nu\nu'$ is not a multiple pipe in $E(H)$).

\medskip
For the purpose of the running time analysis the effect of the operation on $G_\varphi(\nu)$ is that we move the subgraph induced by $B\setminus \{v_1,\ldots, v_k\}$ from $G_\varphi(\nu)$ into a new graph $G_\varphi(\nu')$, and introduce a virtual vertex corresponding to the pipe $\nu\nu'$ in both $G_\varphi(\nu)$ and $G_\varphi(\nu')$, whose degree is $\sum_{i=1}^k\deg_B(v_i)$. We will be often tacitly using the following lemma.

\begin{lemma}
\label{lem:enclose}
Given an instance of atomic embeddability $\varphi$,
an application of \operation{Enclose($B$)} results in an equivalent instance $\varphi':G'\rightarrow H'$.
The operation can be implemented in $O(\sum_{i=1}^k\deg_B(v_i))$ time.
\end{lemma}
\begin{proof}
The equivalence is a consequence of Observation~\ref{obs:atomic}.

For the forward direction, given the set of embeddings $\mathcal{E}_{\mu}$ of  $G_\varphi(\mu)$, $\mu\in V(H)$, inherited from an atomic embedding of $G$, we construct embeddings  $\mathcal{E}_{\mu}'$ of $G_{\varphi'}(\mu)$, $\mu\in V(H')$,
inherited from an atomic embedding of $G'$ as follows.
For $\mu\notin \{\nu,\nu'\}$, we put $\mathcal{E}_\mu'=\mathcal{E}_\mu$.
The embedding of  $\mathcal{E}_{\nu'}'$ is obtained from $\mathcal{E}_\nu|_B$ by identifying $v_1,\ldots, v_k$, which are incident to a common face, thereby turning them into a single virtual vertex corresponding to the pipe $\nu\nu'$.
Finally, $\mathcal{E}_{\nu}'$ is obtained from $\mathcal{E}_\nu$ by contracting $B\setminus \{v_1,\ldots, v_k\}$ into a single virtual vertex corresponding to the pipe $\nu\nu'$.

For the opposite direction, given the set of embeddings $\mathcal{E}_{\mu}'$ of  $G_\varphi(\mu)$, $\mu\in V(H')$, inherited from an atomic embedding of $G'$, we construct embeddings  $\mathcal{E}_{\mu}$ of $G_{\varphi}(\mu)$, $\mu\in V(H)$,
inherited from an atomic embedding of $G$ as follows.
For every $\mu\in V(H)\setminus \{\nu\}$, we put $\mathcal{E}_\mu=\mathcal{E}_\mu'$.
Finally, $\mathcal{E}_{\nu}$ is obtained from the atomic embedding of $G'$ on $\mathcal{S}(\nu)\cup \mathcal{S}(\nu')$ by filling the holes corresponding to pipes, except for $\nu\nu'$, and contracting the fillings to points.
\end{proof}

\begin{figure}[htbp]
\centering
\includegraphics[width=\textwidth]{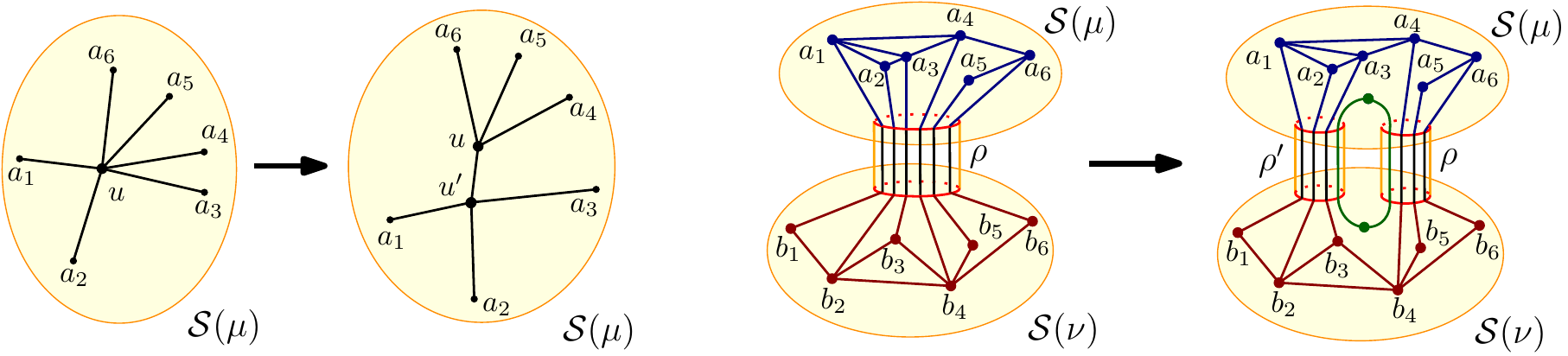}
\caption{An atomic embedding of $G$ on $\mathcal{S}(\nu)$  before and after applying  \operation{Stretch($u,\{uv_1,uv_2,uv_3\}$)}. Vertex $u$ is either ordinary (left) or virtual (right).
If $u$ is a virtual vertex, it corresponds to a pipe $\rho=\mu\nu$.}
\label{fig:rigid}
\end{figure}

In the following we define the operation that replaces a vertex $u$ in $G_\varphi(\mu)$ by an edge $uu'$, and distributes the edges incident to $u$ among $u$ and $u'$. The operation produces an equivalent instance if the rotation of $u$ is fixed, see Fig.~\ref{fig:rigid}.

\paragraph{\operation{Stretch$(u,E_u)$}.}
We are given a vertex $u$ in $G_\varphi(\mu)$ and a set $E_u=\{ua_1,\ldots, ua_\ell\}$ of edges incident to $u$ where $0<|E_u|< \deg(u)$.
We distinguish between two cases, depending on whether $u$ is an ordinary or a virtual vertex.

If $u$ is ordinary, then remove the edges $ua_1,\ldots, ua_\ell$, introduce a new vertex $u'$ and a new edge $uu'$, as well as new edges $u'a_1,\ldots, u'a_\ell$ in $G_\varphi(\mu)$.

If $u$ is virtual, then let $\rho=\mu\nu\in E(H)$ be the pipe corresponding to the virtual vertex $u$ in $G_\varphi(\mu)$ and $v$ in $G_\varphi(\nu)$; and assume that $\original{ua_i}=a_ib_i$, for $i\in [\ell]$,
where  $b_1,\ldots, b_\ell$ are vertices in $G_\varphi(\nu)$. Do the following:
Introduce a new pipe $\rho'=\mu\nu$ in $H$ corresponding to new virtual vertices $u'$ in $G_\varphi(\mu)$ and $v'$ in $G_\varphi(\nu)$;  introduce a new connected component in $G$, which is a cycle with two vertices and two parallel edges $f,f'$ forming a multiple edge such that $\varphi(f)=\rho$ and $\varphi(f')=\rho'$; and finally, modify $\varphi$ by setting $\varphi(a_i b_i)=\rho'$, for all $i\in [\ell]$. In local graphs this corresponds to replacing $ua_i$ and $vb_i$ with $u'a_i$ and $v'a_i$, respectively,
for all $i\in [\ell]$; and inserting two new edges $uv$ and $u'v'$ in the two local graphs, respectively, and subdividing each with an ordinary vertex.

\medskip
For the purposes of the running time analysis (below) the effect of the operation can be seen as the replacement of $u$ by an edge whose two endpoints have degrees $\ell+1$ and $\deg(u)-\ell+1$, respectively (hence the sum of their degrees equals $\deg(u)+2$). If $u$ is a virtual vertex (i.e., corresponds to a pipe between two atoms), then both virtual vertices corresponding to the same pipe go through these changes.
By Observation~\ref{obs:atomic}, the following is straightforward.

\begin{lemma}
\label{lem:stretch}
Given an instance of atomic embeddability $\varphi$ such that the edges in $E_u$ are incident to $u$, and are consecutive in the rotation of vertex $u$ in every embedding of $G_\varphi(\mu)$ inherited from an atomic embedding of $G$, then the operation \operation{Stretch($u,E_u$)} produces an equivalent instance.
\end{lemma}

\begin{corollary}
\label{cor:fixed}
For an instance $\varphi:G\rightarrow H$ of atomic embeddability,
if a vertex $u\in V(G_\varphi(\nu))$ has a fixed rotation,
in which the edges in $E_u$ are consecutive,
then \operation{Stretch($u,E_u$)} produces an equivalent instance.
\end{corollary}

The operation of contraction that follows is applied to an edge $\rho=\nu\mu$ of $H$ and it produces an equivalent instance if each of $G_\varphi(\nu)$ and $G_\varphi(\mu)$ is a p-star or p-path.

\paragraph{\operation{Contract$(\rho)$}.}
We are given a pipe $\rho=\mu\nu$ such that $\rho$ is the only pipe between $\mu$ and $\nu$.
Contract the pipe $\mu\nu$ in $H$ into an atom $\langle\mu\nu\rangle$ and change $\varphi$ accordingly (that is, put $\varphi(u)=\langle\mu\nu\rangle$ for all the vertices mapped by $\varphi$ to $\mu$ or $\nu$).
Let $\varphi'$ denote the resulting instance. Note that $G_{\varphi'}(\langle\mu\nu\rangle)$ might be disconnected,
in which case operation \operation{Split($\langle\mu\nu\rangle$)} is automatically applied
to obtain a normal instance, as explained in Section~\ref{ssec:ds}.
Since $\rho$ is the only pipe between $\mu$ and $\nu$ the operation does not introduce a loop in $H$.

Several incarnations of the following lemma, which is a consequence of Belyi's theorem~\cite{B83_self}, were proved in related papers; see for example, \cite[Lemma~3.2]{AFT19+_weak}, \cite[Claim~7]{FK17_htsocg}, or \cite[Lemma~6]{FKMP15}.

\begin{lemma}
\label{lem:Belyi}
Let $\mu\nu\in E(H)$ be a pipe 
such that
either (i) both $G_\varphi(\mu)$ and $G_\varphi(\nu)$ are p-stars,
or (ii) $G_\varphi(\mu)$ or $G_\varphi(\nu)$ is a p-path;
and in both cases, $\rho$ corresponds to vertices $u$ and $v$ of maximum degree in $G_\varphi(\mu)$ and $G_\varphi(\nu)$, respectively. Then \operation{Contract($\mu\nu$)} produces an equivalent instance $\varphi'$.
\end{lemma}
\begin{proof}
Denote by $\varphi':G\rightarrow H'$ the map returned by \operation{Contract($\mu\nu$)}.
First assume that $\varphi:G\rightarrow H$ is atomic embeddable. Then there exists an atomic embedding $\mathcal{E}:G\rightarrow \mathcal{H}$ (where every vertex $a\in V(G)$ is embedded in $\mathcal{S}(\varphi(a))$; and every edge $ab\in E(G)$ is embedded as a Jordan arc in $\mathcal{S}(\varphi(a))\cup\mathcal{S}(\varphi(b))$ as specified in the definition of atomic embedding). Let $\mathcal{S}(\langle\mu\nu\rangle)=\mathcal{S}(\mu)\cup\mathcal{S}(\nu)$.
Then the thickening $\mathcal{H}'$ of $H'$ equals $\mathcal{H}$, and
the embedding $\mathcal{E}:G\rightarrow \mathcal{H}=\mathcal{H}'$ witnesses that $\varphi':G\rightarrow H'$ is atomic embeddable.

Conversely, assume that $\varphi':G\rightarrow H'$ is atomic embeddable. Let $\mathcal{E}':G\rightarrow \mathcal{H}'$ be an atomic embedding.
Consider the restriction of $\mathcal{E}':G\rightarrow \mathcal{H}'$ on the surface $\mathcal{S}(\langle\mu\nu\rangle)$. Filling the holes of $\mathcal{S}(\langle\mu\nu\rangle)$ with discs, and then contract them to points, to obtain an embedding of $G_{\varphi'}(\langle\mu\nu\rangle)$ on the sphere ${S}^2$.

First, assume that (i) both $G_\varphi(\mu)$ and $G_\varphi(\nu)$ are p-stars:
$G_\varphi(\mu)$ is the union of internally vertex disjoint paths between $u$ and a vertex set $V_a$,
and similarly $G_\varphi(\nu)$ is the union of internally vertex disjoint paths between $v$ and a vertex set $V_b$.
Consequently, $G_{\varphi'}(\langle\mu\nu\rangle)$ is the union of internally vertex disjoint paths between vertices in $V_a$ and $V_b$. (Note that $G_{\varphi'}(\langle\mu\nu\rangle)$ need not be connected.) By suppressing the internal vertices of the paths between $V_a$ and $V_b$, we obtain an embedding of a bipartite multigraph $G^-_{\varphi'}(\langle\mu\nu\rangle)$ with partite sets $V_a$ and $V_b$ on $S^2$.

By Belyi's theorem~\cite{B83_self}, there exists a Jordan curve $\beta:S^1\rightarrow {S}^2$ that intersects every edge of $G^-_{\varphi'}(\langle\mu\nu\rangle)$ in exactly one point, and the intersection is transversal. The curve $\beta$ partitions ${S}^2$ into two parts, $A$ and $B$. We can subdivide the edges of $G^-_{\varphi'}(\langle\mu\nu\rangle)$ to obtain an embedding of $G_{\varphi'}(\langle\mu\nu\rangle)$ on a sphere such that the curve $\beta$ crosses an edge $e\in E(G_{\varphi'}(\langle\mu\nu\rangle))$ if and only if $e\in \varphi^{-1}[\rho]$. Consequently, by contracting $A$ (resp., $B$) into a vertex $v$ (resp., $u$), we obtain an embedding of $G_\varphi(\nu)$ (resp., $G_\varphi(\mu)$) on a sphere, where the vertices $u$ and $v$ have opposite rotations.
Observation~\ref{obs:atomic} now implies that $\varphi:G\rightarrow H$ is atomic embeddable.

Next assume that (ii) $G_\varphi(\mu)$ or $G_\varphi(\nu)$ is a p-path:
Without loss of generality, assume that $G_\varphi(\mu)$ is a p-path, with poles $u$ and $w$.
Consequently, $G_{\varphi'}(\langle\mu\nu\rangle)$  is a subdivision of $G_{\varphi}(\nu)$, obtained by subdividing the edges incident to $v$. In particular $G^-_{\varphi'}(\langle\mu\nu\rangle)$ is isomorphic to $G^-_\varphi(\nu)$,
where vertex $w$ in $G^-_{\varphi'}(\langle\mu\nu\rangle)$ corresponds to vertex $v$ in $G^-_\varphi(\nu)$.
By imposing the rotation of $w$ on $v$ (and the opposite rotation on $u$),
Observation~\ref{obs:atomic} implies that $\varphi:G\rightarrow H$ is atomic embeddable, completing the proof. \end{proof}

\begin{figure}
\centering
\includegraphics[scale=0.8]{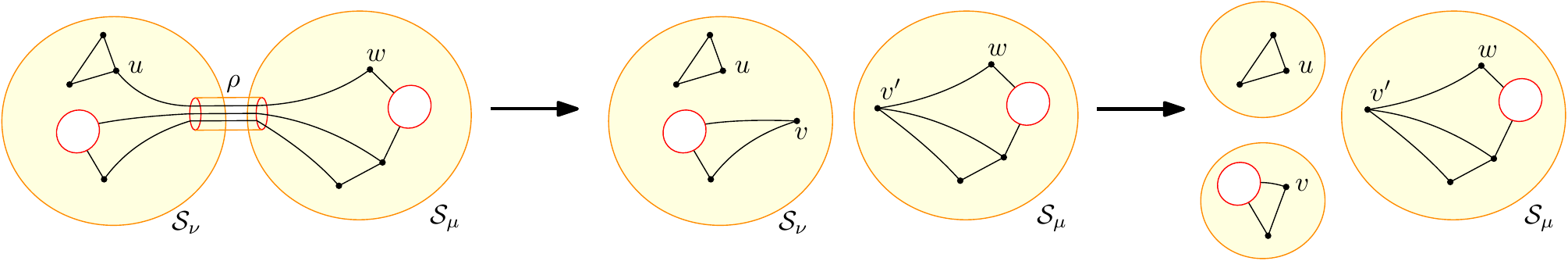}
\caption{An atomic embedding of $G$ on $\mathcal{S}(\nu) \cup \mathcal{S}(\mu)$ before and after operation \operation{Delete($uv$)}, where $uv\in G_\varphi(\nu)$ and $uw=\original{uv}$.  The operation  reduces the degree of a virtual vertex $v$ of $G_\varphi(\nu)$ such that $v$ is incident to a cut-edge, and at most 3 edges of $G$ pass through its corresponding pipe $\rho$.}
\label{fig:delete}
\end{figure}

Our last operation deletes a cut edge of a \emph{subcubic} local graph; see Fig.~\ref{fig:delete} for an illustration.

\paragraph{\operation{Delete$(e)$}.}
We are given a cut edge $e=uv$ in a subcubic local graph $G_\varphi(\nu)$.
If both $u$ and $v$ are ordinary vertices, then delete $uv$ from $E(G)$ (thereby disconnecting $G_\varphi(\nu)$ into two components and invoking \operation{Split($\nu$)}). Else assume w.l.o.g.\ that $u$ is ordinary and $v$ is virtual. Let $\rho=\mu\nu$ be the pipe that corresponds to $v$ in $G_\varphi(\nu)$ and a vertex $v'$ in $G_\varphi(\mu)$, and let $uw\in E(G)$ be the edge corresponding to $uv$, that is, $uw=\original{uv}=\original{v'w}$, where $v'w\in G_\varphi(\mu)$.
Delete the edge $uw$ from $G$ (thereby reducing the degree of $\rho$ to 2), then \operation{Suppress($\rho$)} (which turns $v$ and $v'$ into ordinary vertices), and finally insert an edge $v'w$ into both $G$ and $G_\varphi(\mu)$.

\begin{lemma}
\label{lem:delete-cutedge}
For every instance $\varphi:G\rightarrow H$ of atomic embeddability,
an application of \operation{Delete($e$)} produces an equivalent instance $\varphi'$.
\end{lemma}
\begin{proof}
First, assume that $\varphi$ is atomic embeddable.
If both $u$ and $v$ are ordinary vertices, then the deletion of edge $uv$ produces
an atomic embeddable instance by Observation~\ref{obs:atomic} and
Lemma~\ref{lem:split-atom}.
If $u$ is ordinary and $v$ is virtual, then we can clearly delete $\overline{uv}=uw$ from $G$, suppress the pipe $\rho$ of degree at most 2, and insert the edge $v'w$ (which was already present in $G_{\varphi}(\mu)$ before the operation). By Observation~\ref{obs:atomic}, and Lemmas~\ref{lem:supress-virtual} and~\ref{lem:split-atom}, $\varphi'$ is atomic embeddable.

Conversely, assume that $\varphi'$ is atomic embeddable. Then there exists an atomic embedding $\mathcal{E}':G'\rightarrow \mathcal{H}'$ with respect to $\varphi'$, where $\mathcal{H}'$ is the thickening of $H'$. Note that $\varphi'(u)\neq \varphi'(v)$ since $\varphi'$ is normal and $uv$ was a cut edge in $G_\varphi(\nu)$. Let $\nu_u,\nu_v\in V(H')$ be atoms such that $\varphi'(u)=\nu_u$ and $\varphi'(v)=\nu_v$. Recall that the embedding $\mathcal{E}'$ determines a rotation system on all local graphs of $\varphi'$. Consider disjoint plane embeddings of $G_{\varphi'}(\nu_u)$ and $G_{\varphi'}(\nu_v)$ with the rotation systems inherited from $\mathcal{E}'$ such that $u$ and $v$ are incident to a common face. If $v$ is a virtual vertex of degree 3 in $G_\varphi(\nu)$ (hence $v'$ has degree 3 in $G_{\varphi'}(\mu)$), we can choose plane embeddings of $G_{\varphi'}(\nu_u)$ and $G_{\varphi'}(\nu_v)$ with the additional property that the insertion of the edge $uv$ as a Jordan arc between $u$ and $v$ yields the embedding of $G_\varphi(\nu)$ in which the rotation at $v$ is opposite to the rotation of $v'$ in $G_{\varphi'}(\mu)$. All other local graphs of $\varphi$ are the same as in $\varphi'$, and their rotation systems are inherited from $\mathcal{E}'$. By Observation~\ref{obs:atomic}, $\varphi$ is atomic embeddable.
\end{proof}

\subsection{Toroidal Instance}
\label{ssec:toric}

An instance $\varphi:G\rightarrow H$ is \textbf{toroidal} if $H$ is a cycle and for every atom $\nu\in V(H)$, the graph $G_\varphi(\nu)$ is a p-path in which both poles are virtual vertices, and correspond to the two pipes incident to $\nu$.

Given an instance $\varphi:G\rightarrow H$ and a subgraph $H'\subseteq H$, such that the restriction of $\varphi$ to $G'=\varphi^{-1}[H']$, denoted $\varphi':G'\rightarrow H'$, is toroidal, we say that $H'$ is a \textbf{toroidal cycle} in $H$.

In this section, we show how to decide toroidal instances of atomic embeddability in linear time.
First, note that in a toroidal instance every ordinary vertex has degree 2, hence $G$ is a  disjoint union of cycles, say $C_1,\ldots ,C_t$, for some $t\in \mathbb{N}$. Furthermore, $\varphi$ maps each cycle $C_k$, $k\in[t]$, to a walk that winds around $H$ once or more times.

\begin{lemma}
\label{lem:toroidal}
Let $\varphi:G\rightarrow H$ be a toroidal instance of atomic embeddability, where $H$ is a cycle, and $G$ is a vertex disjoint union of cycles $C_1,\ldots ,C_t$. The instance $\varphi$ is positive if and only if $\varphi(C_k)$ is a walk of the same length for all $k\in[t]$
(that is, every cycle winds around the torus $\mathcal{H}$ the same number of times).
\end{lemma}
Roughly speaking, Lemma~\ref{lem:toroidal} follows by the intersection form of the closed curves on the torus. Indeed, whether a pair of curves could be crossing free on the torus is governed by their  homology classes~\cite[Example~2A.2. and Corollary~3A.6.(b)]{H05_alg}  over~$\mathbb{Z}$.
\begin{proof}
In an atomic embedding, each cycle $C_k$ is embedded on the torus $\mathcal{H}$ as a closed curve, whose homology class over $\mathbb{Z}$ is given by a pair $(i_k,j_k)\in \mathbb{Z}^2$, where we assume w.l.o.g.\ that the first component $i_k$ is the length of the walk $\varphi(C_k)$ divided by $|V(H)|$. In particular, we may assume that $i_k>0$.
Since in an atomic embedding, $C_k$ is mapped to a Jordan curve in $\mathcal{H}$, its homology class must be primitive~\cite{MP78_primitive,M76_primitive}, or in other words, $\gcd(i_k,j_k)=1$, for all $k\in[t]$.

The ``if'' part follows by observing that the restriction of $\varphi$ to a cycle $C_i$  is a positive instance, no matter how many times $\varphi(C_i)$ winds around $H$. Indeed, a desired atomic embedding lies in the primitive homology class $(i_k,1)$.
  A sufficiently small neighborhood of the embedding of $C_i$ in $\mathcal{H}$ is homeomorphic to an annulus, into which all cycles $C_1,\ldots , C_t$ can be embedded
given the hypothesis is satisfied. It remains to prove ``only if'' part.

It is well-known that  the minimum number of crossings between a pair of closed curves in homology classes $(i_k,j_k)$ and $(i_\ell,j_\ell)$ is given as the absolute value of

$$\begin{array}{c}
\begin{pmatrix}
    i_k & j_k
\end{pmatrix}
 \end{array}
\begin{pmatrix}
 0 & -1 \\
 1 & 0
\end{pmatrix}
\begin{pmatrix}
    i_\ell \\ j_\ell
\end{pmatrix}
=i_\ell j_k-i_kj_\ell,$$
which counts the \emph{algebraic intersection} of the pair, see~\cite[Section 6.4.3]{S93_top}, and in particular Exercise~6.4.3.2. therein\footnote{We refrain from properly defining algebraic intersection, since we will not need it in the sequel,
and refer an interested reader to~\cite{S93_top}. To the readers who are unfamiliar with the intersection form it might be only clear that the absolute value of the expression gives a lower bound on the number of crossings between the curves. Nevertheless, for our purpose this lower bound is sufficient, and therefore we do not delve into more details to show that the lower bound is always tight.}.
 Since the cycles $C_1,\ldots, C_t$ embed into $\mathcal{H}$ as pairwise disjoint curves in an atomic embedding, we have $i_\ell j_k-i_kj_\ell=0$ for every pair of distinct $k,\ell\in[t]$. It follows that $\frac{j_k}{i_k}=\frac{j_\ell}{i_\ell}$. Since $\gcd(i_k,j_k)=1$, $\gcd(i_\ell,j_\ell)=1$, $i_k>0$, and $i_\ell>0$, we have that $i_k=i_\ell$ and $j_k=j_\ell$, which concludes the proof.
\end{proof}

\begin{corollary}
\label{cor:toroidal}
We can decide whether a toroidal instance $\varphi:G\rightarrow H$ is atomic embeddable
in time $O(n)$, where $n$ is the number of edges and vertices in $G$.
\end{corollary}
\begin{proof}
Under the assumptions, $G$ is the union of vertex disjoint cycles $C_1,\ldots,C_t$, for some $t\in \mathbb{N}$. We report that the instance is positive if and only if $\varphi(C_i)$ is a walk of the same length for all $k\in[t]$.
This algorithm is correct by Lemma~\ref{lem:toroidal}. It runs in linear time in the size of $G$, as the length of the walks $\varphi(C_i)$, $i\in [t]$, can be computed in a simple traversal of $G$.
\end{proof}

\subsection{The Subcubic Case}
\label{ssec:subcubic}

An instance $\varphi:G\rightarrow H$ of \textsc{atomic embeddability} is \textbf{subcubic} if $G_\varphi(\nu)$ is subcubic (i.e., its maximum degree is at most 3) for every $\nu\in V(H)$. In this section, we show how to decide subcubic instances of \textsc{atomic embeddibility} in linear time.
By Observation~\ref{obs:atomic}, it is enough to check whether all graphs $G_\varphi(\nu)$, $\nu\in V(H)$, are planar, and they each have embeddings in the plane such that for every pipe $\mu\nu\in E(H)$, the virtual vertices corresponding to $\mu\nu$ in $G_\varphi(\mu)$ and $G_\varphi(\nu)$ have opposite rotations.

Planarity testing for a graph takes linear time~\cite{HoTa74_planarity}. Let $n$ be the number of vertices and edges in $G$. Then the disjoint union of all local graphs $\mathcal{G}$ has $O(n)$ size (since each vertex in $V(G)$ corresponds to a unique ordinary vertex, and every edge in $E(G)$ corresponds to one or two edges in $\mathcal{G}$). Hence planarity testing for $\mathcal{G}$ takes $O(n)$ time.

In the subcubic case, every vertex in the local graphs $G_\varphi(\nu)$, $\nu\in V(H)$, has at most two possible rotations (including the vertices of 1- and 2-cuts). We show how to encode the possible embeddings of each local graph by a boolean variable, and then reduce the existence of compatible embeddings to a 2SAT formula, which can be solved in $O(n)$ time.

We start with a postprocessing algorithm that eliminates 1- and 2-edge-cuts from local graphs.

\begin{enumerate}[(1)]
\item[] \textbf{Postprocessing.} We are given a subcubic instance $\varphi:G\rightarrow H$ of atomic embeddability.
\item\label{it:Post1}
While there exists a cut edge $e$ in some $G_\varphi(\nu)$, $\nu\in V(H)$, apply \operation{Delete($e$)}.
\item\label{it:Post2}
While there exists a proper 2-edge-cut $\{e,f\}$ in a local graph of $\varphi$,
such that $e=u_1v_1$ and $f=u_2v_2$, where both $v_1$ and $v_2$ are in a $\{u_1,u_2\}$-bridge $B$ of $G_\varphi(\nu)$, apply \operation{Enclose($B$)} (creating a new pipe $\rho_B$ of degree 2),
and \operation{Suppress($\rho_B$)}.
\end{enumerate}

\begin{lemma}\label{lem:postprocess}
For a subcubic instance $\varphi$ of atomic embeddability of size $n$, Postprocessing runs in $O(n)$ time, and it returns an equivalent subcubic instance $\varphi'$ of size $O(n)$ in which every local graph is a cycle, a p-path, or a subdivided 3-connected planar graph.
\end{lemma}
\begin{proof}
The while loop in Step~\ref{it:Post1} decreases the number of edges in $E(G)$, so it terminates after $O(n)$ iterations, and uses $O(n)$ time. The while loop in Step~\ref{it:Post2} decreases the number of 2-edge-cuts in local graphs, and hence, it also terminates after $O(n)$ iterations in $O(n)$ time. By Lemmas~\ref{lem:supress-virtual}, \ref{lem:delete-cutedge}, and \ref{lem:enclose},
the instance $\varphi':G'\rightarrow H'$ returned by the Postprocessing algorithm is equivalent to $\varphi$.

Step~\ref{it:Post1} can only decrease the degree of a local graph, and Step~\ref{it:Post2} creates a pair of ordinary vertices of degree 2. Since $\varphi$ is a subcubic instance, $\varphi'$ is also subcubic. When the algorithm terminates, then $G^-_{\varphi'}(\nu)$ has neither cut edges nor 2-edge-cuts, for all $\nu\in V(H')$.

We claim that every local graph $G_{\varphi'}(\nu)$, $\nu\in V(H')$, is biconnected. (Note that we assume that $\varphi'$ is normal, and hence, $G_{\varphi'}(\nu)$ is connected.) Let $v$ be an arbitrary vertex in $G_{\varphi'}(\nu)$. Vertex $v$ is incident to at least two edges in each $\{v\}$-bridge, otherwise we would find a cut edge. However, $\deg(v)\leq 3$, so there is at most one $\{v\}$-bridge, and $v$ is not a cut vertex. This completes the proof of the claim.

It remains to show that every local graph $G_{\varphi'}(\nu)$, $\nu\in V(H')$, is a cycle, a p-path, or a subdivided 3-connected graph. Consider a graph $G_{\varphi'}(\nu)$. Suppose it is neither a cycle not a subdivided 3-connected graph. Let $\{u,v\}$ be a proper 2-cut. If $\deg_B(u)=\deg_B(v)=1$, for some $\{u,v\}$-bridge $B$, then $B$ is a $uv$-path, as otherwise the edges incident to $u$ and $v$ in $B$ would form a 2-edge-cut in $G^-_{\varphi'}(\nu)$. The number of $\{u,v\}$-bridges is at least 2 and at most 3, since $G_{\varphi'}(\nu)$ is subcubic. If there are three $\{u,v\}$-bridges, then $\deg_B(u)=\deg_B(v)=1$ for every $\{u,v\}$-bridge $B$, and hence, every $\{u,v\}$-bridge is a $uv$-path, and so $G_{\varphi'}(\nu)$ is a p-path.
Suppose now that there are two $\{u,v\}$-bridges. Then neither bridge can be a $uv$-path, otherwise $\{u,v\}$ would not be a proper 2-cut. Therefore $\max\{\deg_B(u),\deg_B(v)\}\geq 2$ for every $\{u,v\}$-bridge $B$. Since $G_{\varphi'}(\nu)$ is subcubic, this implies $\min\{\deg_B(u),\deg_B(v)\}=1$
for both bridges. By Observation~\ref{obs:proper2edgecut}, there are edges $e=uu'$ and $f=vv'$ in the two bridges that form a 2-edge-cut in $G^-_{\varphi'}(\nu)$, contradicting the assumption that no such 2-edge-cut exists.
\end{proof}

\begin{lemma}
\label{lem:subcubic}
We can decide whether a subcubic instance $\varphi:G\rightarrow H$ is atomic embeddable
in $O(n)$ time, where $n$ is the number of edges and vertices in $G$ and $H$.
\end{lemma}
\begin{proof}
By Lemma~\ref{lem:postprocess}, we may assume that every local graph $G_{\varphi}(\nu)$, $\nu\in V(H)$, is a cycle, a subcubic p-path, or a subdivided 3-connected planar graph. We can ignore cycles, as vertices of degree 2 have only one rotation. Every subdivided 3-connected planar graph has two possible rotation systems given by an embedding, that are equivalent up to a reflection.
In every embedding of a p-path with poles $u$ and $v$, the cyclic order of the $\{u,v\}$-bridges around $u$ and $v$ are reverse of each other. If $\deg(u)=\deg(v)=3$, the three $\{u,v\}$-bridges have two possible cyclic orders,
that is, the p-path have two possible rotation systems in a plane embedding.

For every local graph $G_{\varphi}(\nu)$, $\nu\in V(H)$, that is a p-path or a subdivided 3-connected graph, we introduce a boolean variable $x_\nu$, which is the indicator variable for the two possible rotation systems of $G_{\varphi}(\nu)$ in a plane embedding. In other words, we fix an embedding of $G_\varphi(\nu)$ corresponding to $x_\nu=1$, and then the reflected embedding  corresponds to $x_\nu=0$. Since $\varphi$ is a normal instance, every pipe has degree 3, that is, it corresponds to two virtual vertices of degree 3 in two local graphs. In particular, for every pipe $\rho=\mu\nu$, both $x_\mu$ and $x_\nu$ are defined. For every pipe $\rho=\mu\nu$, we introduce  a constraint $x_\mu=x_\nu$ if the rotations of its two corresponding virtual vertices in $G_{\varphi}(\mu)$ and $G_{\varphi}(\nu)$ in the embedding of $G_{\varphi}(\mu)$ and $G_{\varphi}(\nu)$, respectively, corresponding to $x_\nu=1$ and $x_\mu=1$ are opposite of each other; and $x_\mu=\neg x_\nu$ otherwise.

These constraints yield an instance of 2SAT with $O(|V(H)|)$ boolean variables and $O(|E(H)|)$ constraints, which can be solved in $O(n)$ time. If the 2SAT instance is positive then the graphs $G_\varphi(\nu)$, $\nu\in V(H)$, each have an embedding in the plane such that for every pipe $\mu\nu\in E(H)$, the corresponding virtual vertices in $G_\varphi(\mu)$ and $G_\varphi(\nu)$ have opposite rotations. It follows that $\varphi$ is a positive instance by Observation~\ref{obs:atomic}. Conversely, if $\varphi$ is a positive instance, then the atomic embedding of $G$ induces plane embeddings of the local graphs such that for every pipe $\mu\nu\in E(H)$, the corresponding virtual vertices have opposite rotations. By construction, the indicator variables $x_\nu$, $\nu\in V(H)$, satisfy all constraints of the 2SAT instance.
\end{proof}

\subsection{Main Algorithm}
\label{ssec:main}

We define two subroutines and then present our main algorithm. Subroutine~1 ensures that our instance has some desirable properties, and Subroutine~2 decreases the maximum degree over all local graphs $G_\varphi(\nu)$, for all atoms $\nu\in V(H)$, that are not contained in a toroidal cycle $C$ of $H$.

%

The crucial part of our algorithm reduces the maximum degree in local graphs over all atoms that are not in toroidal cycles. Specifically, for an instance $\varphi$ of atomic embeddability, let
\begin{itemize}
\item $V^*(H)$ (resp., $E^*(H)$) be the set of atoms (resp., pipes) in $H$ that are \textbf{not} in any toroidal cycle of $H$; and
\item let $\Delta(\varphi)$ be the maximum degree over all vertices of all local graphs $G_\varphi(\nu)$, $\nu\in V^*(H)$, if $V^*(H)\neq\emptyset$, and let $\Delta(\varphi)=2$ if $V(H)=\emptyset$.
\end{itemize}

We first call Subroutine~1 for a normal instance $\varphi$, and show that it returns an equivalent instance in which the proper 1- and 2-cuts in local graphs $G_\varphi(\nu)$, $\nu\in V^*(H)$, are in a special form, as described in terms of the following definition.

\begin{definition}
An instance $\varphi$ of atomic embeddability is \textbf{$d$-nice}, for $d\geq 3$, if it meets the following two conditions:
\begin{enumerate}[nosep,label=(N\arabic*)]
\item\label{N0} $\Delta(\varphi)\leq d$.
\item\label{N1} If $\deg(v)= d$  for a vertex $v$ of some local graph $G_\varphi(\nu)$, $\nu\in V(H)$,
     then $v$ has a fixed rotation, or $G_\varphi(\nu)$ is a p-path or a p-star.
\item\label{N2} If $\rho=\mu\nu\in E(H)$, such that both $G_\varphi(\mu)$ and $G_\varphi(\nu)$ are p-stars, and $\rho$ corresponds to virtual vertices of degree at least $d$, then $\rho$ is the only pipe between $\mu$ and $\nu$.
\end{enumerate}
\end{definition}

\begin{figure}
\centering
\includegraphics[scale=0.8]{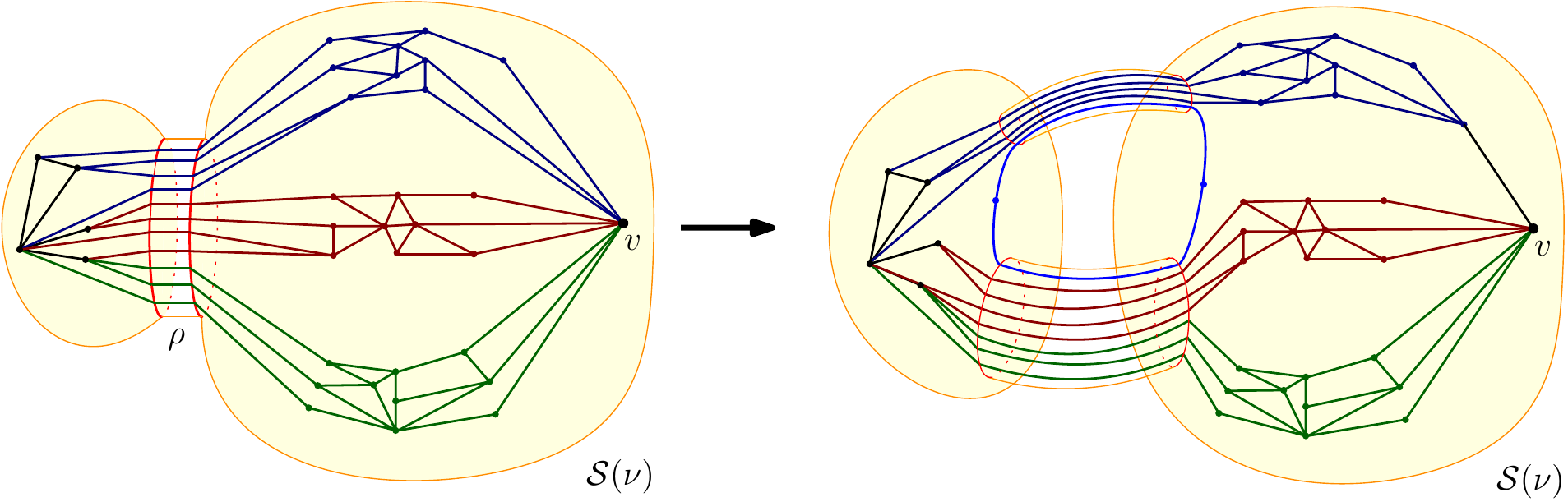}
\caption{An atomic embedding of $G$ on $\mathcal{S}(\nu)$  before and after Step~\ref{it:Max0},
where $u$ is a virtual vertex of $G_\varphi(\nu)$ corresponding to the pipe $\rho$, and $v$ is an ordinary vertex of $G_\varphi(\nu)$.}
\label{fig:2cut}
\end{figure}

We present a subroutine that takes a normal instance $\varphi$ as an input, and returns a $\Delta$-nice normal instance for the maximum degree $\Delta$ over all local graphs, that is, $\Delta=\Delta(\varphi)$ (as shown in Lemma~\ref{lem:S1} and Corollary~\ref{cor:S1-runtime} below).

\begin{enumerate}[label=(\roman*),series=steps]
\item[] \textbf{Subroutine~1}. Input: a normal instance $\varphi$ of atomic embeddability, where $\Delta(\varphi)\geq 4$.

\item\label{it:Max0}
While there is a proper 2-cut $\{u,v\}$ and a nonseparable $\{u,v\}$-bridge $B$ in   $G_\varphi(\nu)$, for some $\nu \in V(H)$, such that $\max\{\deg(u),\deg(v)\}=\Delta$, but neither $u$ nor $v$ is a cut vertex, then do the following:
Perform \operation{Stretch($u,E_u$)}, where $E_u$ is the set of edges in $E(B)$ incident to $u$; and
perform \operation{Stretch($v,E_v$)}, where $E_v$ is the set of edges in $E(B)$ incident to $v$.
If $u$ or $v$ is a virtual vertex corresponding to a pipe $\mu\nu$ and $G_{\varphi}(\mu)$ is nonplanar,
report that the instance $\varphi$ is not atomic embeddable and exit the subroutine.

\item\label{it:SubI3}
While there is a proper 2-edge-cut $\{e,f\}$ in  $G_\varphi(\nu)$, for some $\nu \in V(H)$,
then let $e=u_1v_1$ and $f=u_2v_2$ such that both $v_1$ and $v_2$ are in a $\{u_1,u_2\}$-bridge $B$ of $G_\varphi(\nu)$, then apply \operation{Enclose($B$)} (creating a new pipe $\rho_B$ of degree 2),
and \operation{Suppress($\rho_B$)}.

\item\label{it:Max1}
While there is a proper cut vertex $v$ with $\deg(v)=\Delta$ in some local graph of $\varphi$,
then successively apply \operation{Enclose($B$)} for every bridge $B$ of $v$
 (thereby turning every bridge of $v$ into a p-path).
 Apply \operation{Suppress($\rho_B$)} if applicable.
\end{enumerate}

In Section~\ref{ssec:runtime} (cf.~Corollary~\ref{cor:S1-runtime}), we show that Subroutine~1 terminates and analyse its running time. Here we prove that if it terminates, it returns a $\Delta(\varphi)$-nice instance.

\begin{lemma}\label{lem:S1}
For an instance $\varphi$ of atomic embeddibility, if Subroutine~1 terminates,
it either returns an equivalent, normal, and $\Delta(\varphi)$-nice instance $\varphi'$,
or reports that $\varphi$ is not atomic embeddable.
\end{lemma}
\begin{proof}
Let $\Delta=\Delta(\varphi)$ for short. By Lemmas~\ref{lem:supress-virtual}, \ref{lem:enclose}, and~\ref{lem:stretch}, Subroutine~1 returns an equivalent instance $\varphi'$ upon termination. Note that instance $\varphi'$ is normal, since we apply \operation{Suppress($\nu_B$)} to any pipe of degree less than 3. The operations in Subroutine~1 do not increase the maximum degree in any local graph outside of toroidal cycles; and make no changes at all in local graphs in toroidal cycles. Consequently, $\Delta(\varphi')\leq \Delta$.

At the end of Step~\ref{it:Max0}, every $\{u,v\}$-bridge is separable for every proper $\{u,v\}$-cut where  $\max\{\deg(u),\deg(v)\}=\Delta$; see Fig.~\ref{fig:2cut}. We consider Step~\ref{it:SubI3} now. Suppose that $\{u,v\}$ is a proper 2-cut in $G_\varphi(\nu)$, such that $\min\{\deg(u),\deg(v)\}\ge 3$ and $\max\{\deg(u),\deg(v)\}=\Delta$

If there exist exactly two (separable) $\{u,v\}$-bridges in $G_\varphi(\nu)$ (none of which is a subdivided edge as otherwise $\{u,v\}$ would not be a proper 2-cut), then Step~\ref{it:SubI3} eliminates the proper 2-cut $\{u,v\}$ by a single application of \operation{Enclose($.$)}, due to Observation~\ref{obs:proper2edgecut}, and does not introduce any new proper 2-cut. Indeed, up to symmetry there are two cases to consider depending on whether $\deg_B(u)=1$ or $\deg_B(u)=\deg(u)-1$, and $\deg_B(v)=1$ or $\deg_B(v)=\deg(v)-1$; see Fig.~\ref{fig:2edgecut}.

\begin{figure}
\centering
\includegraphics[width=\textwidth]{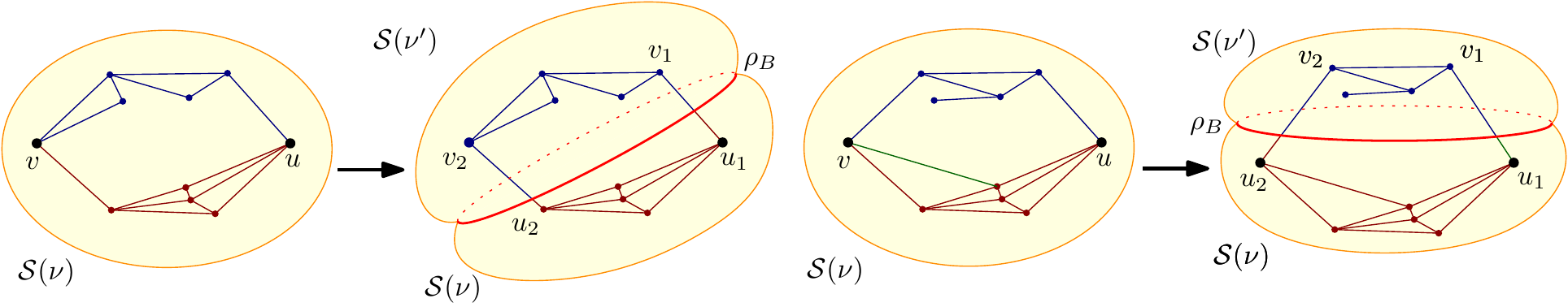}
\caption{The operation \operation{Enclose($B$)} in Step~\ref{it:SubI3} eliminates every proper 2-cut $\{u,v\}$ in $G_\varphi(\nu)$ that has exactly two $\{u,v\}$-bridges. Either one of the $\{u,v\}$-bridges is separable (left), or there exists a separable $\{u_1,u_2\}$-bridge such that $\deg_B(u_1)=\deg_B(u_2)=1$ (right).}
\label{fig:2edgecut}
\end{figure}

If there exist at least three (separable) $\{u,v\}$-bridges in $G_\varphi(\nu)$ such that  $\max\{\deg(u),\deg(v)\}=\Delta$, then Step~\ref{it:SubI3} turns $G_\varphi(\nu)$ into a p-path with the poles $u$ and $v$.

Hence, at the end of Step~\ref{it:SubI3}, for every proper 2-cut $\{u,v\}$ we have
(a) $\max\{\deg(u),\deg(v)\}<\Delta$; or
(b) $u$ or $v$ is a cut vertex of degree $\Delta$; or
(c) $u$ and $v$ are the poles of a p-path.
In particular, every vertex $w$ with $\deg(w)=\Delta$ in a local graph $G_\varphi(\mu)$,
is a proper 1-cut, or a pole of a p-path, or has fixed rotation.

Step~\ref{it:Max1} successively turns every cut vertex of degree $\Delta$ into the center of a p-star. It creates new 2-cuts within these p-stars and possibly in adjacent atoms, but it does not create any new vertex of degree $\Delta$. Hence, at the end of Subroutine~1
$\varphi$ satisfies~\ref{N0} and \ref{N1}. For property~\ref{N2}, note that by enclosing all the bridges of the center of every p-star in $\mathcal{G}$ of degree $\Delta$, Step~\ref{it:Max1} eliminates possible problematic multiple pipes $\rho=\mu\nu$ in $H$, where $\rho$ corresponds to a pair of centers of p-stars $G_\varphi(\mu)$ and $G_\varphi(\nu)$. Overall, the instance $\varphi'$ returned by Subroutine~1 upon termination meets conditions~\ref{N0}--\ref{N2}, consequently $\varphi'$ is $\Delta$-nice.
\end{proof}

\begin{figure}
\centering
\includegraphics[scale=1]{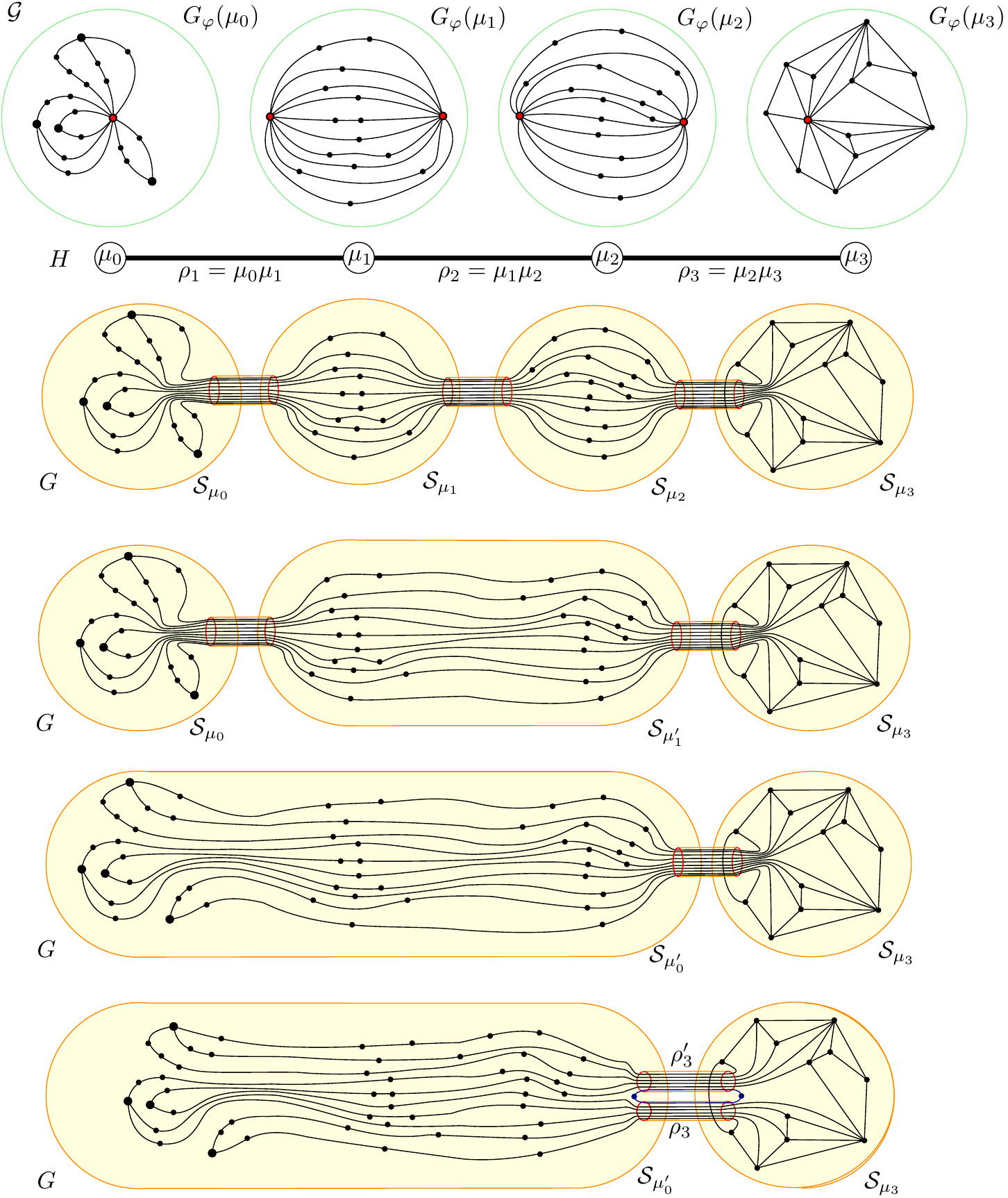}
\caption{A path $(\rho_1,\rho_2,\rho_3)$ in $H$ through the atoms $\mu_0,\ldots, \mu_3$. The graph $G_\varphi(\mu_0)$ is a p-star, $G_\varphi(\mu_1)$ and $G_\varphi(\mu_2)$ are p-paths, and $G_\varphi(\mu_3)$ is 3-connected. In each local graph, the virtual vertices corresponding to $\rho_1$, $\rho_2$, or $\rho_3$ are vertices of maximum degree. The bottom three subfigures show the effect of Step~\ref{it:Max2a} and Step~\eqref{it:Max3c} of Subroutine~2 on the graph $G$ in this instance.}
\label{fig:useful}
\end{figure}

\begin{figure}
\centering
\includegraphics[scale=1]{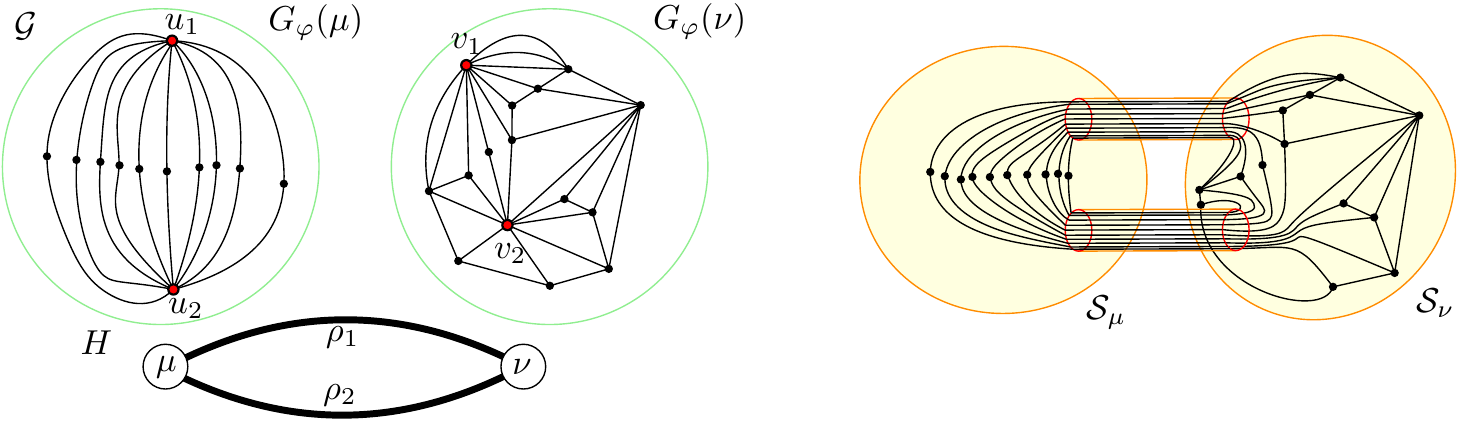}
\caption{An illustration of the setting in Step~\ref{it:Max2b}, analogous to Fig.~\ref{fig:useful}.}
\label{fig:useful++}
\end{figure}

\paragraph{Degree reduction.} We are now ready to present the crucial subroutine of our algorithm that
reduces $\Delta(\varphi)$ by eliminating all vertices of degree $\Delta(\varphi)$ that are not in toroidal cycles.
(See Figs.~\ref{fig:useful} and~\ref{fig:useful++} for the possible relations between virtual vertices of degree $\Delta(\varphi)$.)

\begin{enumerate}[label=(\roman*),resume=steps]
\item[] \textbf{Subroutine~2}. Input: a normal and $\Delta$-nice instance $\varphi$ of atomic embeddability,
    where $\Delta=\Delta(\varphi)$ and $\Delta\geq 4$.
\item\label{it:Max2}
While there exists a pipe $\mu\nu\in E^*(H)$ of degree $\Delta$ such that $G_\varphi(\mu)$ or $G_\varphi(\nu)$ is a p-path,  do the following. Suppose w.l.o.g.\ that $G_\varphi(\mu)$ is a p-path.
  \begin{enumerate}[label=(\alph*),ref=(\roman{enumi}.\alph*)]
  \item\label{it:Max2a}
  If $\mu\nu$ is not a multiple edge apply \operation{Contract($\mu\nu$)}.
  \item\label{it:Max2b}
  Else there exists a pair of pipes $\rho_1$ and $\rho_2$ joining $\mu$ with $\nu$.
  Let $u_i$ and $v_i$, resp., be virtual vertices in $G_\varphi(\mu)$ and $G_\varphi(\nu)$ corresponding to $\rho_i$ for $i\in \{1,2\}$. Note that $u_1$ and $u_2$ are the poles of the p-path $G_\varphi(\mu)$; and both $v_1$ and $v_2$ are fixed due to~\ref{N1}, as $\mu$ is not in a toroidal cycle. Apply \operation{Stretch($v_1,E_1$)} and \operation{Stretch($v_2,E_2$)}, where $E_i$ is a set of $\lfloor\Delta/2\rfloor$ consecutive edges in the rotation at $v_i$, for $i\in\{1,2\}$.
  If $G_\varphi(\mu)$ becomes nonplanar, report that the instance is not atomic embeddable and exit the subroutine.

    \end{enumerate}
\item\label{it:Max3}
For every pipe $\mu\nu\in E^*(H)$ of degree $\Delta$ that corresponds to virtual vertices $u$ and $v$
in $G_\varphi(\mu)$ and $G_\varphi(\nu)$, respectively, do:
     \begin{enumerate}[label=(\alph*),ref=(\roman{enumi}.\alph*)]
     \item\label{it:Max3a}
     If both $u$ and $v$ have fixed rotations (in $G_\varphi(\mu)$ and $G_\varphi(\nu)$, resp.), then
     check whether the two rotations are compatible. If they are incompatible, then report that $\varphi$ is not atomic embeddable and exit the subroutine. Otherwise apply \operation{Stretch($u,E_u$)}, where $E_u$ is a set of $\lfloor\Delta/2\rfloor$ consecutive edges in the rotation of $u$.
     \item\label{it:Max3b}
     If neither $u$ nor $v$ has a fixed rotation, then apply \operation{Contract($\mu\nu$)}. This contracts $\mu\nu$ into a new atom, denoted by $\langle\mu\nu\rangle$, and combines $G_\varphi(\mu)$ and $G_\varphi(\nu)$ into a new graph $G_\varphi(\langle\mu\nu\rangle)$.
     If $G_\varphi(\langle\mu\nu\rangle)$ is nonplanar, report that $\varphi$ is not atomic embeddable and exit the subroutine.
     \item\label{it:Max3c}
     Else assume w.l.o.g.\ that $u$ has fixed rotation in $G_\varphi(\mu)$, and is incident to edges $(uv_1,\ldots, uv_\Delta)$ in this cyclic rotation order. Successively apply \operation{Stretch($u,.$)}, turning vertex $u$ into an induced binary tree with $\Delta-2$ vertices.
     If $G_\varphi(\nu)$ is nonplanar, report that the instance is not atomic embeddable and exit the subroutine.
     \end{enumerate}
\item \label{it:Max4}
For every ordinary vertex $v\in V(G_\varphi(\nu))$, $\nu\in V(H)$, with $\deg(v)=\Delta$ that has fixed rotation, apply \operation{Stretch($v,E_v$)}, where $E_v$ is a set of $\lfloor\Delta/2\rfloor$ consecutive edges in the rotation of $v$.
\item \label{it:Max5}
For every ordinary vertex $v\in V(G_\varphi(\nu))$, $\nu\in V(H)$, with $\deg(v)=\Delta$ that is part of a 1- or 2-cut, apply \operation{Detach($v$)}.
\end{enumerate}

This completes the description of Subroutine~2. In Section~\ref{ssec:runtime} we show that Subroutine~2 terminates and analyse its running time. In Lemma~\ref{lem:S2} below, we prove that if it terminates, it returns an instance $\varphi'$ with $\Delta(\varphi')< \Delta(\varphi)$. We first clarify when an operation \operation{Stretch($.$)} can create a proper 1- or 2-cut.

\begin{lemma}
\label{lem:decontraction}
Let $u$ be a vertex in $G_\varphi(\nu)$ such that $\deg(u)\ge 4$, and assume that operation \operation{Stretch($u,.$)} produces an instance $\varphi'$ in which $u$ is replaced by an edge $uu'$.
If $u$ is not a proper 1-cut in $G_\varphi(\nu)$, then neither $u$ nor $u'$ is a proper 1-cut in $G_{\varphi'}(\nu)$.
If $u$ neither is a proper 1-cut nor belongs to a proper 2-cut in $G_\varphi(\nu)$, then neither $u$ nor $u'$ belongs to a proper 2-cut in $G_{\varphi'}(\nu)$.
\end{lemma}
\begin{proof}
For the sake of contradiction  suppose w.l.o.g.\ that $u'$ is a proper 1-cut in $G_{\varphi'}(\nu)$. If we contract the edge $uu'$ back into the vertex $u$, the vertex $u$ is a proper 1-cut in $G_{\varphi}(\nu)$. Indeed, if $B$ is the $\{u'\}$-bridge containing $u$, then $B$ contains at least 3 edges since $\deg(u)\ge 3$. Thus, the bridge $B$ gives rise to a $\{u\}$-bridge after we contract $uu'$.

Similarly, suppose that $u'$ or $u$ belongs to a proper 2-cut $\{v,w\}$ of $G_{\varphi'}(\nu)$. If we contract the edge $uu'$ back into the vertex $u$, no loops are created, since $uu'$ is not a multiple edge. We show that the vertex $u$ belongs to a proper 2-cut in $G_{\varphi}(\nu)$ or is a proper 1-cut. The proper 2-cut $\{v',w'\}$ or 1-cut $\{v'\}$ in $G_{\varphi}(\nu)$ is constructed as follows. We put $v'=u$ and $w'=u$ if $v\in\{u,u'\}$ and $w\in \{u,u'\}$, respectively,
and we put $v'=v$ and $w'=w$, otherwise.
The proper 2-cut $\{v',w'\}$ is, in fact, a proper 1-cut if $v'=w'=u$.

To see that $\{v',w'\}$ is a proper 1-cut if $v'=w'$, we observe that $vw=uu'$ and apply the definitions of a proper 1-cut and proper 2-cut. Otherwise, w.l.o.g.\ $v=u$, and thus, $v'=u$. Let $B$ denote the $\{v,w\}$-bridge containing $uu'$ edge. The bridge $B$ is not a (subdivided) edge, since $\deg(u')\ge 3$ and $w\neq u'$.  Thus, after contracting $uu'$, the bridge $B$ gives rise to at least one $\{v',w'\}$-bridge that is not a subdivided edge.  Hence, $\{v',w'\}$ is a proper 2-cut.
\end{proof}

\begin{lemma}\label{lem:S2}
For a $\Delta(\varphi)$-nice instance $\varphi:G\rightarrow H$ of atomic embeddability, if Subroutine~2 terminates, then it either returns an equivalent normal instance $\varphi':G'\rightarrow H'$ such that $\Delta(\varphi')< \Delta(\varphi)$, or reports that $\varphi$ is not atomic embeddable.
\end{lemma}
\begin{proof}
Let $\Delta=\Delta(\varphi)$.
We show that every step of Subroutine~2 maintains a $\Delta$-nice normal instance equivalent to $\varphi$ until it terminates; and it either returns such an instance $\varphi'$ or reports that $\varphi'$ is not atomic embeddable. Subroutine~2 maintains a normal instance until termination, since it does not create virtual vertices of degree 2 and \operation{Split($.$)} is applied automatically whenever a local graph disconnects into two or more components.

Step~\ref{it:Max2a} produces an equivalent instance by Lemma~\ref{lem:Belyi}: Each \operation{Contract(.)} operation merges a local graph $G_\varphi(\nu)$ with a p-path $G_\varphi(\mu)$ and produces a new local graph $G_\varphi(\langle\mu\nu\rangle)$, where $G^-_\varphi(\langle\mu\nu\rangle)$ is isomorphic to $G^-_\varphi(\nu)$, so the instance remains $\Delta$-nice.
In Step~\ref{it:Max2b}, the two invocation of \operation{Stretch($.$)} produce an equivalent instance by Lemma~\ref{lem:stretch}. At the end of Step~\ref{it:Max2}, none of the local graphs outside of toroidal cycles is a p-path of degree $\Delta$.

In Step~\ref{it:Max3a}, the rotation of the virtual vertices $u$ and $v$ must be compatible in any atomic embedding by Observation~\ref{obs:atomic}. If they are compatible, then operation \operation{Stretch($u,.$)} produces an equivalent instance by Corollary~\ref{cor:fixed}; and the resulting instance is still $\Delta$-nice since no new proper 1-cut or 2-cut is introduced in $G_\varphi(\mu)$ and $G_\varphi(\nu)$ by Lemma~\ref{lem:decontraction}.

In Step~\ref{it:Max3b}, both $G_\varphi(\mu)$ and $G_\varphi(\nu)$ are p-star, centered at $u$, and $v$ resp., since instance is $\Delta$-nice,
and p-paths of degree $\Delta$ have already been eliminated. By Lemma~\ref{lem:Belyi}, \operation{Contract$(\mu\nu)$} produces an equivalent instance. The resulting instance is still $\Delta$-nice, since the maximum degree of $G_\varphi(\langle \mu\nu\rangle)$ is less than $\Delta$; and $G_\varphi(\langle \mu\nu\rangle)$ is planar if $\varphi$ is atomic embeddable by Observation~\ref{obs:atomic}.

In Step~\ref{it:Max3c}, operation \operation{Stretch($u,.$)} for a fixed vertex $u$ yields an equivalent instance by Corollary~\ref{cor:fixed}. If the rotation of $u$ is incompatible with a p-star centered at $v$, \operation{Stretch($u,.$)} may turn $G_\varphi(\nu)$ into a nonplanar graph, and then $\varphi$ is not atomic embeddable by Observation~\ref{obs:atomic}. Note that $G_\varphi(\nu)$ is a p-star, since p-paths have been eliminated in Step~\ref{it:Max2}. Successive \operation{Stretch($.$)} operations eliminate the only vertex of degree $\Delta$ of $G_\varphi(\nu)$, namely $v$. The graph $G_\varphi(\mu)^-$ remains 3-connected by Lemma~\ref{lem:decontraction}, and so the resulting instance is $\Delta$-nice.

The equivalence of Step~\ref{it:Max4} follows by Lemma~\ref{lem:stretch}, and the resulting instance is still $\Delta$-nice, since operation \operation{Stretch($v,.$)} does not introduce a 1-cut in $G_\varphi(\nu)$ that would violate the $\Delta$-nice property by Lemma~\ref{lem:decontraction}.

In Step~\ref{it:Max5}, operation \operation{Detach($u$)} is applied in a $\Delta$-nice instance, hence it produces an equivalent instance by Lemma~\ref{lem:split}, which obviously remains $\Delta$-nice.

For the remainder of the proof, assume that Subroutine~2 returns an instance $\varphi'$. The operations in Subroutine~2 do not increase the maximum degree outside of toroidal cycles; and make no changes within toroidal cycles. Consequently, $\Delta(\varphi')\leq \Delta$.

It remains to prove that $\Delta(\varphi')< \Delta$, i.e., that Subroutine~2 eliminates vertices of degree $\Delta$ from all local graphs outside of toroidal cycles; call these vertices \textbf{$\Delta$-critical}. Since $\varphi$ is $\Delta$-nice, every $\Delta$-critical vertex in a local graph has a fixed rotation, or it is a center of a p-star, or a pole of a p-path. Steps~\ref{it:Max2}--\ref{it:Max5} each eliminate one or two $\Delta$-critical vertices (possibly a pair of corresponding virtual vertices), and do not create any new $\Delta$-critical vertices.
Steps~\ref{it:Max2}--\ref{it:Max3} eliminate all possible $\Delta$-critical virtual vertices; and Steps~\ref{it:Max4}--\ref{it:Max5} eliminate all $\Delta$-critical ordinary vertices. Since Subroutine~2 maintains a $\Delta$-nice instance,  it ultimately eliminates all $\Delta$-critical vertices, and so $\Delta(\varphi')<\Delta$, as claimed.
\end{proof}

\paragraph{Algorithm.} We are given a normal instance $\varphi$ of atomic embeddability.
\begin{enumerate}[(I)]
\item\label{it:Alg1}
While $\Delta(\varphi)\geq 4$, do the following.
    \begin{enumerate}[label=(\alph*),ref=(\Roman{enumi}.\alph*)]
    \item\label{it:loop}  Call Subroutine~1 (which turns $\varphi$ into a $\Delta$-nice instance) followed by Subroutine~2 (which reduces $\Delta(\varphi)$). If Subroutine~1 or Subroutine~2 reports that the instance $\varphi$ is not atomic embeddable return \texttt{False} and terminate the algorithm.
    \end{enumerate}
\item\label{it:Alg2}
For each connected component $C$ of $H$, let $G(C)=\varphi^{-1}[C]$.
    \begin{enumerate}[label=(\alph*),ref=(\Roman{enumi}.\alph*)]
    \item If $C$ is a toroidal cycle of $H$, decide atomic embeddability for $\varphi|_{G(C)}$ using Corollary~\ref{cor:toroidal}.
    \item Else decide atomic embeddability for $\varphi|_{G(C)}$ using Lemma~\ref{lem:subcubic}.
    \end{enumerate}
\item\label{it:Alg3}
If $\varphi|_{G(C)}$ is atomic embeddable for all components $C$ of $H$, then return \texttt{True}; else return \texttt{False}.
\end{enumerate}

In Section~\ref{ssec:runtime} we show that the Algorithm terminates and analyse its running time.
Here we show that if it terminates it correctly decides the atomic embeddability problem.

\begin{lemma}\label{lem:Alg}
Suppose that Algorithm terminates for an instance $\varphi:G\rightarrow H$. Then the algorithm returns \texttt{True} if and only if $\varphi$ is atomic embeddable.
\end{lemma}
\begin{proof}
Since the input $\varphi$ is normal, it is a valid input for Subroutine~1 in the first iteration of Step~\ref{it:loop}. By Lemma~\ref{lem:S1}, Subroutine~1 returns a $\Delta(\varphi)$-nice instance and therefore the input for Subroutine~2 is valid. In any subsequent iteration of Step~\ref{it:loop}, Subroutine~1 receives a valid input as Subroutine~2 returns a normal instance $\varphi^*$ (however, this instance need not be $\Delta(\varphi^*)$-nice). By Lemma~\ref{lem:S1} and Lemma~\ref{lem:S2}, the while loop in Step~\ref{it:Alg1} terminates after at most $\Delta(\varphi)-3$ iterations, and returns an equivalent instance (or correctly reports that $\varphi$ is not atomic embeddable).
If the Algorithm proceeds to Steps~\ref{it:Alg2}--\ref{it:Alg3}, the correctness of the output follows from Corollary~\ref{cor:toroidal} and Lemma~\ref{lem:subcubic}.
\end{proof}

\subsection{Running Time Analysis}
\label{ssec:runtime}

\paragraph{Potential Functions.}
We measure the progress of the algorithm, for an instance $\varphi:G\rightarrow H$, using three parameters defined as follows. Recall that $\mathcal{G}$ denotes the disjoint union of all local graphs $G_\varphi(\nu)$, $\nu\in V(H)$.
\begin{itemize}\itemsep 0pt
\item Let $N(\varphi)=|V(\mathcal{G})|$, that is, be the number of vertices of $\mathcal{G}$.
\item let $N_{\geq 3}(\varphi)=|\{v\in V(\mathcal{G}): \deg(v)\geq 3\}|$, i.e., the number of vertices of $\mathcal{G}$ of degree 3 or higher.
\item Let the \textbf{potential} of $\varphi$ be

$$\Phi(\varphi)=\sum_{v\in V(\mathcal{G})}  \left(\max\{0,\deg(v)-\xi(v)\}\right)^{\sigma(v)},$$
where $\xi(v)=2$ and $\sigma(v)=3$ if $v$ is a proper cut vertex, $\xi(v)=2$  and  $\sigma(v)=2$ if $v$ is part of a proper 2-cut but not a cut vertex, and $\xi(v)=3$ and $\sigma(v)=1$ otherwise.
\end{itemize}
Note that $\max\{0,\deg(v)-2\}=0$ if  $\deg(v)\leq 2$, that is, the vertices of local graphs of degree less than 3 do not contribute to the potential. Clearly $\deg(v)<N(\varphi)$ for every $v\in V(\mathcal{G})$,
and so $\Phi(\varphi)\leq N^4(\varphi)$ is a trivial upper bound. Our analysis hinges on the following charging scheme:

\paragraph{Overview.} We show below (Lemma~\ref{lem:MaxDegree}) that each iteration of Step~\ref{it:loop} of the Algorithm decreases the potential. This readily implies that the while loop in Step~\ref{it:Alg1} terminates (hence the Algorithm terminates, which completes the proof of correctness).
Recall that Step~\ref{it:loop} runs Subroutines~1 and~2, that is, it applies Steps~\ref{it:Max0}--\ref{it:Max5}.
We show that both the number of elementary operations performed and the number of new vertices created in Steps~\ref{it:Max0}--\ref{it:Max5} are bounded from above by a constant times the decrease of the potential.
Step~\ref{it:SubI3} does not change the potential, so we need additional machinery to bound its running time: We use the parameters $N(\varphi)$ and $N_{\geq 3}(\varphi)$. We continue with the specifics.

\paragraph{Analysis.}
Recall that each iteration of the while loop of Step~\ref{it:Alg1} of the Algorithm,
which calls Subroutine~1 followed by Subroutine~2. The two subroutine jointly perform
Steps~\ref{it:Max0}--\ref{it:Max5}. We use the following notation. Assume that $\varphi_0$ is the input of Subroutine~1,
and we obtain instances $\varphi_1,\ldots , \varphi_7$ at the end of Step~\ref{it:Max0},$\ldots$,\ref{it:Max5}.
Denote by $\mathcal{G}_i$ the union of all local graphs in the instance $\varphi_i$ for $i=1,\ldots, 7$.
The following lemma is helpful for the analysis of Step~\ref{it:Max0}.

\begin{lemma}\label{lem:isolate2Bridge}
Let $\{u,v\}$ be a proper 2-cut in a local graph $G_\varphi(\nu)$ such that $\max\{\deg(u),\deg(v)\}\ge 4$;
and let $B$ be a nonseparable $\{u,v\}$-bridge.
Then one iteration of the while loop in Step~\ref{it:Max0} produces an instance $\varphi'$ such that
$\Phi(\varphi')<\Phi(\varphi)$ and
$N(\varphi')\leq N(\varphi)+ 4$.
\end{lemma}
\begin{proof}
If $u$ (resp., $v$ is ordinary), then operation \operation{Stretch(.)} creates $1$ new vertex in $\mathcal{G}$;
and if it is virtual, it creates at most $2$ new vertices in $\mathcal{G}$. Overall at most $4$ new vertices are created, all of which are of degree 3 or higher. That is, $N(\varphi')\leq N(\varphi)+ 4$.

In the following we analyze how the operation impacts the degree of vertices.
If $u$ is an ordinary vertex, then $\operation{Stretch(.)}$ changes its degree
from $\deg(u)$ to $\deg(u)-\deg_B(u)+1\leq \deg(u)$;
and creates a new vertex of degree $\deg_B(u)+1\leq \deg(u)$.
If $u$ is a virtual vertex, corresponding to a pipe $\rho$, then
both virtual vertices corresponding to $\rho$ go through the same changes.

By Lemma~\ref{lem:decontraction}, the new vertices created by \operation{Stretch(.)}
are not cut vertices (although they may participate in proper 2-cuts).
If $u$ is a virtual vertex, then the corresponding virtual vertex $w$ in an adjacent atom $\mu$
might be a cut vertex or a vertex in a proper 2-cut.
By Lemma~\ref{lem:decontraction}, if $w$ or a new vertex $w'$ created by \operation{Stretch($u,.$)}
in $G_{\varphi'}(\mu)$ is a proper cut vertex (resp., contained in of a proper 2-cut),
then so is $w$ in $G_{\varphi}(\mu)$. In particular, for any
existing vertex $z$ in $G_\varphi(\mu)$ (including vertex $w$),
the exponent $\sigma(z)$ cannot increase, $\xi(z)$ cannot decrease.
Consequently, the contribution of $z$ to the potential cannot increase.

The change in the potential incurred by $u$ is most
$$(\deg_B(u)-1)^2+(\deg(u)-\deg_B(u)-1)^2-(\deg(u)-2)^2,$$
which is nonpositive by convexity, and equals to zero if and only if $\deg_B(u)=1$.
Since $B$ is a nonseparable $\{u,v\}$-bridge, we have $\max\{\deg_B(u),\deg_B(v)\}\geq 2$,
and the two \operation{Stretch(.)} operations at $u$ and $v$ jointly decrease the potential.
\end{proof}

\begin{corollary}\label{cor:Max0}
Let $\varphi_0$ be a normal instance of atomic embeddability and let $\Delta=\Delta(\varphi_0)$.
Then the while loop in Step~\ref{it:Max0} terminates and returns an instance $\varphi_1$,
after at most $\Phi(\varphi_0)-\Phi(\varphi_1)$ iterations,
such that
$N(\varphi_1)\leq N(\varphi_0)+12(\Phi(\varphi_0)-\Phi(\varphi_1))$.
\end{corollary}
\begin{proof}
By Lemma~\ref{lem:isolate2Bridge}, each iteration of the while loop in Step~\ref{it:Max0} decreases the potential.
Consequently, the while loop terminates, and performs at most $\Phi(\varphi)-\Phi(\varphi')$ iterations.
Each iteration applies up to two \operation{Stretch(.)} operations, at $u$ or $v$ for some proper $\{u,v\}$-cut,
and increases the number of vertices by at most $4$.
By Lemma~\ref{lem:isolate2Bridge}, the number of vertices increases by at most 4 times the decrease of the potential.
Summation over all iterations of the while loop in Step~\ref{it:Max0} yields
$N(\varphi_1)\leq N(\varphi_0)+4(\Phi(\varphi_0)-\Phi(\varphi_1))$.
\end{proof}

We can now focus on Steps~\ref{it:SubI3}--\ref{it:Max5}.

\begin{lemma}\label{lem:SubI3}
Let $\varphi_1$ be an instance returned by Step~\ref{it:Max0}.
Then the while loop in Step~\ref{it:SubI3} terminates after at most $N_{\geq 3}(\varphi_1)$ iterations,
and it returns an instance $\varphi_2$ such that
$N(\varphi_2)<N(\varphi_1)+2N_{\geq 3}(\varphi_1)$,
$N_{\geq 3}(\varphi_2)=N_{\geq 3}(\varphi_1)$, and
$\Phi(\varphi_2)=\Phi(\varphi_1)$.
\end{lemma}
\begin{proof}
Let $\varphi$ be an instance at the beginning of one iteration of the while loop in Step~\ref{it:SubI3}.
Since $\varphi$ is normal, every local graph $G_\varphi(\nu)$ is connected. Let $\{e,f\}$  be a proper 2-edge-cut in $G_\varphi(\nu)$, let $e=u_1v_1$ and $f=u_2v_2$ such that both $v_1$ and $v_2$ are in a $\{u_1,u_2\}$-bridge $B$.
Note that each component of $G_\varphi(\nu)\setminus \{e,f\}$ contains a vertex that has degree at least 3 in $G_\varphi(\nu)$, otherwise one of the components would be a path, and the 2-edge-cut would not be proper.
In one iteration of Step~\ref{it:SubI3}, an operation \operation{Enclose(.)} replaces $G_\varphi(\nu)$ with two local graphs obtained by removing edges $e$ and $f$, and inserting two new paths $(u_1,w_u,u_2)$ and $(v_1,w_v,v_2)$, where $w_u$ and $w_v$ are new ordinary vertices of degree 2. In particular, $\Phi$ and the number of vertices of degree at least 3 do not change, and the total number of vertices in $\mathcal{G}$ increases by 2.

It follows that $\Phi(\varphi_2)=\Phi(\varphi_1)$ and $N_{\geq 3}(\varphi_2)=N_{\geq 3}(\varphi_1)$.
Since each iteration in the while loop of Step~\ref{it:SubI3} increases the number of components of $\mathcal{G}$, but each new component contains at least one vertex of degree 3 or higher, the number of iterations is at most $N_{\geq 3}(\varphi)-1$.
Summation over all iterations yields $N(\varphi_2)\leq N(\varphi_1)+2N_{\geq 3}(\varphi_1)-2$.
\end{proof}

Now we are ready to show that Subroutine~1 terminates.
\begin{corollary}\label{cor:S1-runtime}
For an instance $\varphi_0$ of atomic embeddibility of size $n$, Subroutine~1 terminates.
\end{corollary}
\begin{proof}
 By Corollary~\ref{cor:Max0} the while loop in Step~\ref{it:Max0} terminates.
By Lemma~\ref{lem:SubI3}, Step~\ref{it:SubI3} terminates and eliminates all proper 2-edge-cuts containing an edge that is incident to a vertex of degree $\Delta=\Delta(\varphi_0)$. Finally, in the while loop of Step~\ref{it:Max1}, each iteration decreases the number of vertices of degree $\Delta$ in local graphs that are not p-stars. Therefore this while loop terminates, as well.
\end{proof}

Note that Step~\ref{it:Max1} increases the potential. We analyse the combined effect of Steps~\ref{it:Max1}--\ref{it:Max5}, and show that they jointly decrease the potential, and we can charge the number of operations, as well as the number of new vertices to the decrease of the potential. The following observation will be helpful.

\begin{lemma}\label{lem:degrees2}
Let $d$ and $d_1,\ldots ,d_k$ be positive integers such that $d=\sum_{i=1}^k d_i$, $d\geq 4$, and $k\geq 2$.
Then we have
\begin{equation}\label{eq:deg5}
(d-2)^3\geq (2d-5)+\sum_{i=1}^k (\max\{0,d_i-2\})^2 + \sum_{i=1}^k (\max\{0,d_i-2\})^3.
\end{equation}
\end{lemma}
\begin{proof}
We distinguish among three cases. In Case~1, $k=d$ (hence $d_i=1$ for all $i\in[k]$).
Then the right hand side of \eqref{eq:deg5} is less than $2(d-2)$. Clearly, $2(d-2)\leq(d-2)^3$ for $d\geq 4$.

In Case~2, we assume that $2<k<d$. First note that for every $i\in [k]$, we have $d_i-1\ge \max\{0,d_i-2\}$,
and so $(d_i-1)^3\geq (\max\{0,d_i-2\})^3+ (\max\{0,d_i-2\})^2$. Elementary calculation yields

\begin{align}
(d-2)^3 =&\left(\left(\sum_{i=1}^k d_i\right)-2\right)^3
        =\left(\left(\sum_{i=1}^k d_i\right)-k+k-2\right)^3
        =\left((k-2)+\sum_{i=1}^k (d_i-1)\right)^3 \nonumber\\
        =&(k-2)^3+\sum_{i=1}^k (d_i-1)^3+ 3(k-2)^2(d-k)+3(k-2)\sum_{i=1}^k \left(d_i-1\right)^2+ \nonumber \\
         &+6(k-2)\sum_{i=1}^k(d_i-1)(d-d_i-k+1)
          +\sum_{i=1}^k(d_i-1)\sum_{j=1}^k(d_j-1)(d-d_j-k+1) \nonumber\\
        \geq& \sum_{i=1}^k (d_i-1)^3+ 3(k-2)^2(d-k)
        \geq  \sum_{i=1}^k (\max\{0,d_i-2\})^3+ \sum_{i=1}^k(\max\{0,d_i-2\})^2 + (2d-5), \nonumber
\end{align}
where we have dropped some nonegative terms, and used the inequality $3(k-2)^2(d-k)\geq 3(d-3)\geq (2d-5)+(d-4)\geq 2d-5$ for $d\geq 4$.

In Case~3, assume $k=2$. Then, using $d\geq 4$ again, we have

\begin{align}
(d-2)^3 &= \left(\sum_{i=1}^2 (d_i-1)\right)^3
        \geq \sum_{i=1}^2(d_i-1)^3+ 3\sum_{i=1}^2 (d_i-1)= \sum_{i=1}^2(d_i-1)^3+ 3(d-2) \nonumber\\
        &>  \sum_{i=1}^k (\max\{0,d_i-2\})^3+ \sum_{i=1}^k(\max\{0,d_i-2\})^2 + (2d-5), \nonumber
\end{align}
as claimed.
\end{proof}

\begin{lemma}\label{lem:MaxDegree}
Consider Steps~\ref{it:Max1}--\ref{it:Max5} in an invocation of Subroutine~1 followed by Subroutine~2.
We have $\Phi(\varphi_2)>\Phi(\varphi_7)$, $N(\varphi_7)\leq N(\varphi_2)+ O(\Phi(\varphi_2)-\Phi(\varphi_7))$,
and the number of operations performed in Steps~\ref{it:Max1}--\ref{it:Max5} is $O(\Phi(\varphi_2)-\Phi(\varphi_7))$.
\end{lemma}
\begin{proof}
None of these steps increases the number of vertices of degree $\Delta$ or higher in local graphs. Ultimately all
vertices of degree $\Delta$ outside of toroidal cycles are eliminated.

\paragraph{Overview.}
Each operation in Steps~\ref{it:Max1}--\ref{it:Max5} is associated to either a unique vertex of degree $\Delta$,
or two virtual vertices of degree $\Delta$ that correspond to the same pipe.
In Step~\ref{it:Max1} and Steps~\ref{it:Max4}--\ref{it:Max5}, this is vertex $v$;
in Steps~\ref{it:Max2}--\ref{it:Max3}, these are virtual vertices $u$ and $v$ corresponding to the pipe $\mu\nu$.
We consider each vertex $v$ of degree $\Delta$ in the instance $\varphi_2$, and analyse how the operations associated with $v$ change the potential and the total number of vertices over Steps~\ref{it:Max1}--\ref{it:Max5}.
Let $D(\Phi,v)$ and $D(N,v)$, resp., denote the changes in $\Phi(.)$ and $N(.)$ incurred by the operations associated with vertex $v$.
We claim that for every vertex $v$ of degree $\Delta$ in $\mathcal{G}_2$, we have
\begin{equation}\label{eq:charge1}
D(\Phi,v)\leq 0,
\end{equation}
with equality if and only if $v$ is a a local graph $G_{\varphi_2}(\nu)$ where $\nu$ is in a toroidal cycle; and

\begin{equation}\label{eq:charge2}
D(N,v)+20\ D(\Phi,v)\leq 0.
\end{equation}
Note that \eqref{eq:charge1} holds with a strict inequality for at least one vertex $v$. Indeed, we have $\Delta=\Delta(\varphi_3)=\Delta(\varphi_2)$, and so there is a vertex of degree $\Delta$ in some local graph of
$\varphi_2$ outside of toroidal cycles. Summation over all vertices of degree $\Delta$ then yields

\begin{align}
\Phi(\varphi_7)&=\Phi(\varphi_2)+\sum_{u\in V(\mathcal{G}_{2}):\deg(u)=\Delta} D(\Phi,u)< \Phi(\varphi_2),\nonumber\\
N(\varphi_7)   &=N(\varphi_2)+\sum_{u\in V(\mathcal{G}_{2}):\deg(u)=\Delta} D(N,u)\leq N(\varphi_2)+20(\Phi(\varphi_2)-\Phi(\varphi_7)).\nonumber
\end{align}

\paragraph{Elimination of p-paths.} Recall that each iteration of Step~\ref{it:Max2a} applies \operation{Contract($\mu\nu$)} on a pipe $\mu\nu$ corresponding to virtual vertices $u$ in $G_\varphi(\mu)$
and $v$ in $G_\varphi(\nu)$. Without loss of generality, assume that $G_\varphi(\mu)$ is a p-path with poles $u$ and $w$. Operation \operation{Contract($\mu\nu$)} eliminates $u$ and $v$, and creates a new local graph $G_\varphi(\langle\mu\nu\rangle)$ where $G^-_\varphi(\langle\mu\nu\rangle)$ is isomorphic to $G_\varphi(\nu)$.
For the analysis of $D(\Phi,.)$ and $D(N,.)$, we assume that this operation eliminates $u$ and $w$; and
vertex $v$ of $G_\varphi(\nu)$ survives in $G_\varphi(\langle\mu\nu\rangle)$. Thus, the effect of \operation{Contract($\mu\nu$)} is neutral for $v$, although $v$ may become an ordinary vertex if $w$ is ordinary before the operation.

Inequalities \eqref{eq:charge1}--\eqref{eq:charge2} clearly hold for any vertex $v$ in toroidal cycles.
For all other vertices of degree $\Delta$, we distinguish between three cases as follows.

\paragraph{Vertices of fixed rotation.}
Let $v$ be a vertex of fixed rotation with $\deg(v)=\Delta=\Delta(\varphi_2)$ in $\varphi_2$.
If $v$ is an ordinary vertex, then Steps~\ref{it:Max0}--\ref{it:Max4} do not change $v$,
and in Step~\ref{it:Max4} a \operation{Stretch($v,.$)} operation replaces $v$ with two vertices
$v_1$ and $v_2$, where $\deg(v_1)+\deg(v_2)=\Delta+2$, and $\min\{\deg(v_1),\deg(v_2)\}\geq 3$.
In this case, $D(\Phi,v)=(\deg(v_1)-3)+(\deg(v_2)-3)-(\Delta-3)=-1$, and $D(N,v)=1$.
If $v$ is a virtual vertex, then $\deg(v)$ decreases in either Step~\ref{it:Max2b},~\ref{it:Max3a} or~\ref{it:Max3c}.
In Step~\ref{it:Max2b} or~\ref{it:Max3a} one \operation{Stretch(.)} operation has the same effect on the potential
as for ordinary vertices, $D(\Phi,v)=-1$, but it creates two new vertices, and so $D(N,v)=2$.
In Step~\ref{it:Max3c}, $\Delta-3$ successive  \operation{Stretch(.)} operations
replace $v$ with $\Delta-2$ vertices of fixed orientation with degree 3.
Thus, $D(\Phi,v)=0-(\Delta-2)$, and $D(N,v)\leq \Delta-3$.
In all cases, \eqref{eq:charge1}--\eqref{eq:charge2} follow.

\paragraph{Poles of p-paths.}
Let $v$ be a pole of a p-path $G_{\varphi_2}(\nu)$, with $\deg(u)=\Delta=\Delta(\varphi_2)$.
Denote the other pole of the p-path by $u$, where obviously $\deg(u)=\deg(v)=\Delta$.
If both $u$ and $v$ are ordinary, then Steps~\ref{it:Max0}--\ref{it:Max3} do not change $G_{\varphi_2}(\nu)$.
The \operation{Detach(.)} operation in Step~\ref{it:Max5} replaces $v$
with $\Delta$ new vertices of degree 1. Thus, $D(\Phi,v)=-(\Delta-2)^2$
and $D(N,v)=\Delta-1$, thus \eqref{eq:charge1}--\eqref{eq:charge2} follow.

Assume that $u$ or $v$ is a virtual vertex. Then a \operation{Contract(.)} operation in Step~\ref{it:Max2a} eliminates both $u$ and $v$. We have $D(\Phi,u)=-(\Delta-2)^2$ and $D(N,u)=-1$.

Step~\ref{it:Max2b} applies \operation{Stretch($.$)} to a pair of virtual vertices $u_1$ and $u_2$ of $G_\varphi(\mu)$. Thus, we have $D(\Phi,u_1)=D(\Phi,u_2)<0$ and $D(N,u_1)=D(N,u_2)=1$.

\paragraph{Proper cut vertices.}
Let $v$ be a proper cut vertex in $\varphi_2$ with $\deg(v)=\Delta=\Delta(\phi_2)$ in some local graph $G_{\varphi_2}(\nu)$. Assume $v$ has $k\geq 2$ bridges $B_1,\ldots , B_k$, and $\deg_{B_i}(v)=d_i$ for all $i\in [k]$. Step~\ref{it:Max1} successively encloses the $k$ bridges. Note that $v$ remains a proper cut vertex of degree $\Delta$. Step~\ref{it:Max1} creates new virtual vertices $v_1,\ldots ,v_k$ in the p-star centered at $v$, where $\deg(v)=\sum_{i=1}^k \deg(v_i)$. Every new virtual vertex $v_i$, $i\in [k]$, is part of a proper 2-cut $\{v,v_i\}$.

Moreover, every virtual vertex $v_i$, $i\in [k]$, corresponds to another virtual vertex $v_i'$ in the local graph of an atom created by enclosing $B_i$; this local graph is isomorphic to $B_i$, where $v_i'$ plays the role of $v$. In particular $v_i'$ cannot be a cut vertex, but it may be contained in a proper 2-cut. At the end of Step~\ref{it:Max1}, we obtain a $\Delta$-nice instance $\varphi_3$ in which none of the local graphs containing a virtual vertex $v_i'$, $i\in [k]$, is a p-path or a p-star contain any vertex of degree $\Delta$. Therefore Steps~\ref{it:Max2}--\ref{it:Max5} do not change the degree of $v_i'$, and $v_i'$ cannot become a cut vertex for any $i\in [k]$.

Next we consider the possible changes to the p-star centered at $v$ in Steps~\ref{it:Max2}--\ref{it:Max5}.
Step~\ref{it:Max2a} may turn $v$ into an ordinary vertex as noted above (but it changes neither $D(\Phi,v)$ nor $D(N,v)$). Since $v$ is a proper cut vertex of degree $\Delta$, the next step that can possibly modify the p-star is Step~\ref{it:Max3b}, \ref{it:Max3c}, or~\ref{it:Max5}. In Step~\ref{it:Max3b}, a \operation{Contract(.)} operation eliminates vertex $v$, and any vertex $v_i$, $i\in [k]$ may become a cut vertex.
In Step~\ref{it:Max3c}, successive \operation{Stretch(.)} operations replace $v$ with $\Delta-2$ vertices of degree 3.
Since each of these vertices could be a proper cut vertex, they contribute $(\Delta-2)(3-2)^3=\Delta-2$ to the potential.
Finally, if $v$ is an ordinary vertex, then \operation{Detach($v$)} in Step~\ref{it:Max5} replaces $v$ with $\Delta$ vertices of degree 1, which do not contribute to the potential.

At the beginning of Step~\ref{it:Max1}, vertex $v$ contributes $(\Delta-2)^3$ to $\Phi(\varphi_2)$.
At the end of Subroutine~2, the contribution of $v$, together with the virtual vertices $v_i$ and $v_i'$, over all $i\in[k]$, is at most

$$(\Delta-2)+\sum_{i=1}^k(\max\{0,d_i\})^3+\sum_{i=1}^k(\max\{0,d_i\})^2.$$
By Lemma~\ref{lem:degrees2}, $D(\Phi,v)\leq -(\Delta-3)\leq -\Delta/4$.

Let us estimate the number of new vertices created in these steps.
In Step~\ref{it:Max1}, the \operation{Enclose(.)} operations create
a pair of virtual vertices for each bridge of $v$ (i.e., $2k$ vertices),
and up to $k$ ordinary subdivision vertices.
In Step~\ref{it:Max3c}, \operation{Stretch($.$)} operations create $\Delta-2$ new vertices;
and in Step~\ref{it:Max5}, the \operation{Detach($v$)} operation
increases the number of vertices by $\Delta-1$. Therefore, $D(N,v)\leq 3k+2\Delta-3\leq 5\Delta-3$.
Since $D(\Phi,v)\leq -\Delta/4$, inequalities \eqref{eq:charge1} and \eqref{eq:charge2} follow.
\end{proof}

\begin{lemma}\label{lem:charge}
For an instance $\varphi_0$ of atomic embeddability of size $n$, Algorithm terminates, it performs
$O(N_{\geq 3}(\varphi_0)+\Phi(\varphi_0))$ operations, and runs in $O(n^8)$ time.\footnote{Optimizing the running time analysis further, which we believe is possible, is beyond the scope of this work.}
\end{lemma}
\begin{proof}
Consider one iteration of the while loop of Step~\ref{it:Alg1}, which calls Subroutines~1 and~2.
By Corollary~\ref{cor:Max0}, Step~\ref{it:Max0} terminates, performs $O(\Phi(\varphi_0)-\Phi(\varphi_1))$ operations, and returns and instance $\varphi_1$ with and $N(\varphi_1)\leq N(\varphi_0)+12(\Phi(\varphi_0)-\Phi(\varphi_1))$.
By Lemma~\ref{lem:SubI3}, Step~\ref{it:SubI3} terminates, performs $O(N_{\geq 3}(\varphi_0))$ operations, and returns an instance $\varphi_2$
with $N_{\geq 3}(\varphi_2)=N_{\geq 3}(\varphi_1)$ and $\Phi(\varphi_2)=\Phi(\varphi_1)$.
Similarly, by Lemma~\ref{lem:MaxDegree}, the sequence of Steps~\ref{it:Max1}--\ref{it:Max5} terminates,
performs $O(\Phi(\varphi_2)-\Phi(\varphi_7))$ operations, and returns an instance $\varphi_7$ with
$N(\varphi_7)\leq N(\varphi_2)+O(\Phi(\varphi_2)-\Phi(\varphi_7))$.

Using the definition of the potential, we can bound its initial value by
$$\Phi(\varphi_0)=\sum_{v\in V(\mathcal{G}_0)} (\max\{0,\deg(v)-\xi(v)\})^{\sigma(v)}\leq n\cdot (\Delta(\varphi_0)-2)^3\leq O(n^4).$$
The while loop in Step~\ref{it:Alg1} of the Algorithm terminates after $\Delta(\varphi_0)-3\leq n$ iterations,
since each iteration decreases $\Delta(.)$ by Lemmas~\ref{lem:S1} and~\ref{lem:S2}.
In each iteration, the potential $\Phi(.)$ decreases, and $N(.)$ increases by at most constant times the decrease of the potential
by Lemmas~\ref{lem:isolate2Bridge},~\ref{lem:SubI3} and~\ref{lem:MaxDegree}.
In particular, for every instance $\varphi^*$ in intermediate phases of Step~\ref{it:Alg1},
both $N(\varphi^*)$ and $N_{\geq 3}(\varphi^*)$ are bounded by $O(n+\Phi(\varphi_0))\leq O(n^4)$.

Each operation in Steps~\ref{it:Max0}--\ref{it:Max5} can be implemented in $O(N({\varphi^*}))$ time, where $\varphi^*$ is the instance for which the operation is applied (this allows for planarity testing, and recomputing block trees and SPQR-trees after  each operation). As noted above, we have $N({\varphi^*})\leq O(n+\Phi(\varphi_0))\leq O(n^4)$. The overall running time of all invocations of Step~\ref{it:Max0}--\ref{it:Max5} is $O(n^4 (n+\Phi(\varphi_0)))\leq O(n^8)$.

By Lemmas~\ref{lem:toroidal} and~\ref{lem:subcubic}, Steps~\ref{it:Alg2}--\ref{it:Alg3}
of the Algorithm run in $O(N(\varphi_7))\leq O(n+\Phi(\varphi_0))\leq O(n^4)$ time.
\end{proof}

\begin{theorem}
\label{thm:mainAtom}
There is an algorithm that determines whether a simplicial map $\varphi:G\rightarrow H$ is atomic embeddable
in time polynomial in the number of edges and vertices in $G$ and $H$.
\end{theorem}
\begin{proof}
Let $\varphi$ be an instance $\varphi$ of atomic embeddability of size $n$, where $n$ is the number of edges and vertices in $G$ and $H$. The Preprocessing algorithm runs in $O(n)$ time and returns an equivalent normal instance $\varphi_0$ of size $O(n)$ by Lemma~\ref{lem:Pre}. The main Algorithm for $\varphi_0$ terminates in $O(n^8)$ time by Lemma~\ref{lem:charge}, and determines whether $\varphi_0$ is atomic embeddable by Lemma~\ref{lem:Alg}. Since $\varphi_0$ and $\varphi$ are equivalent, this also determines whether $\varphi$ is atomic embeddable.
\end{proof}

\section{Beyond Atomic Embeddings}
\label{sec:compatible}

Since atomic embeddability is tractable, it makes sense to consider
its generalizations in which every atom can have genus higher than 0.

We consider the \textsc{generalized atomic embedding} problem for a simplical map $\varphi:G\rightarrow H$, where $G$ and $H$ are multigraphs without loops. The only difference from \textsc{atomic embeddibility} is that we define the surface $\mathcal{H}$ as follows: For each atom $\nu\in V(H)$, we construct $\mathcal{S}(\nu)$ from an oriented surface of genus $g(\nu)$ without boundary (rather than a 2-sphere), and remove $\deg(\nu)$ holes. Hence, an \textbf{instance} for generalized atomic embeddability is a pair $(\varphi, g)$, where $\varphi$ is a simplicial map $\varphi: G\rightarrow H$ and $g:V(H)\rightarrow \mathbb{N}_0$.

\begin{problem}[\textsc{Generalized atomic embeddability}]\label{prob:gatomic}
Given a simplicial map $\varphi: G \rightarrow H$, where $G$ and $H$ are multigraphs without loops,
and a function $g:V(H)\rightarrow \mathbb{N}_0$,
decide whether a generalized atomic embedding of $G$ with respect to $\varphi$ exists.
\end{problem}

In this section, we show that \textsc{generalized atomic embeddability} is NP-hard, and therefore also NP-complete, even when $g(\nu)\leq 1$ for all atoms $\nu\in V(H)$, and the number of vertices in $\varphi^{-1}[\nu]$ is at most 7 for each atom $\nu$ with $g(\nu)=1$.

\begin{figure}
\centering
\includegraphics[scale=1]{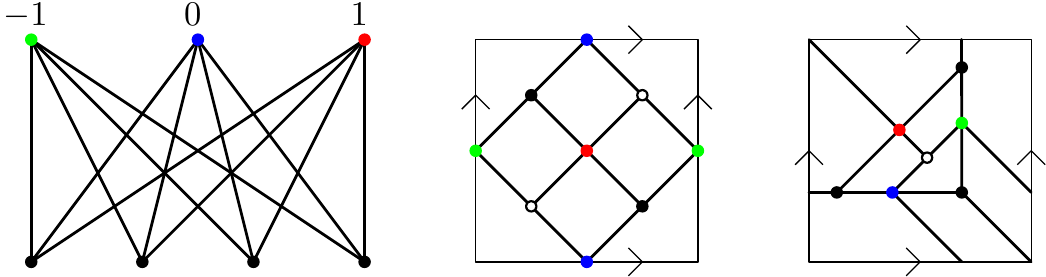}
\caption{The complete bipartite graph $K_{3,4}$ (left) and  its toroidal embeddings (middle and right) in which 4 cubic vertices do not have the same rotation. The torus is obtained by identifying the pairs of opposite sides of the square as indicated by arrows.}
\label{fig:k34}
\end{figure}

The NP-hardness proof is based on the embeddings of $K_{3,4}$ on a torus. For an embedding of $K_{3,4}$ on a torus,
we say that two vertices $u$ and $v$ of the same vertex class (i.e., with the same degree) have \textbf{the same rotation}
if the rotation of $u$ is $(uv_1,\ldots ,uv_k)$ and the rotation at $v$ is $(vv_1,\ldots,vv_k)$, where
$\{v_1,\ldots , v_k\}$ is a vertex class of $K_{3,4}$.

\begin{lemma}
\label{lem:k34}
In every embedding of $K_{3,4}$ on the torus the four cubic vertices do not all have the same rotation;
subject to the previous claim, any rotations for the  four cubic vertices can be realized by a toroidal embedding.
\end{lemma}

\begin{proof}
The first author and Kyn\v{c}l~\cite[Theorem 7(b)]{FK17_counter}
recently proved that in every embedding of $K_{3,4}$ on the torus
there exist two cubic vertices that do not have the same rotation.

It remains to prove that in the following two cases there exists an embedding of $K_{3,4}$ on the torus.
In one case, exactly 2 of the 4 cubic vertices have the same rotation; and
in the other case, exactly 3 of the 4 cubic vertices have the same rotation.
Desired toroidal embeddings are given in Fig.~\ref{fig:k34}(middle) and (right).
\end{proof}

\begin{figure}
\centering
\includegraphics[scale=1]{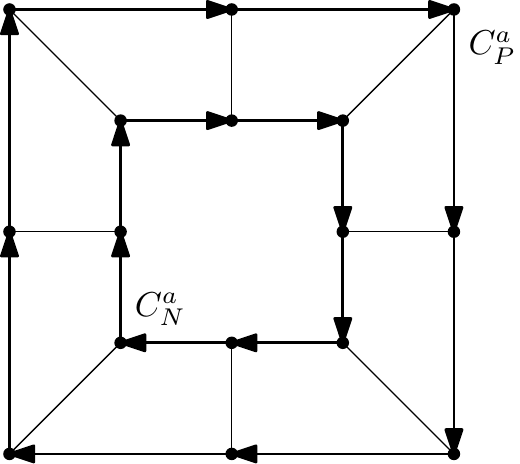}
\caption{The variable gadget $G_0^a=C^a\square P_1$ for the NP-hardness reduction of \textsc{Not-All-Equal 3SAT} to \textsc{Generalized atomic embeddability}.}
\label{fig:var}
\end{figure}

\begin{theorem}
\label{thm:main3}
\textsc{Generalized Atomic Embeddability} is NP-hard.
\end{theorem}

\begin{proof}
We reduce \textsc{Generalized Atomic Embeddability} from \textsc{Not-All-Equal 3SAT}, which is known to be NP-complete.
An instance of \textsc{Not-All-Equal 3SAT} is given by a pair $(\mathcal{A}, \mathcal{C})$,
where $\mathcal{A}$ is a finite set of boolean variables and $\mathcal{C}$ is a finite set of clauses,
each of which is a conjunction of three literals. Each literal is either a variable $a\in \mathcal{A}$ or the negation of $a$, denoted by $\neg{a}$. An instance $(\mathcal{A}, \mathcal{C})$ is \textbf{positive} if there exists an assignment
$\tau: \mathcal{A}\rightarrow \{\texttt{true},\texttt{false}\}$ such that at least 1 and at most 2 literals
are \texttt{true} in every clause.

Given an instance $(\mathcal{A}, \mathcal{C})$ of \textsc{Not-All-Equal 3SAT},
we construct an instance $(\varphi,g)$ for \textsc{Generalized atomic embeddability}, and
show that it is positive if and only if $(\mathcal{A}, \mathcal{C})$  is positive.
Let $\mathcal{C}=\{C_1,\ldots, C_n\}$.

We construct an instance $(\varphi,g)$, where $\varphi: G \rightarrow H$, $g:V(H)\rightarrow \mathbb{N}_0$.
Let the multigraph $H$ be a p-star with a center $\nu_0$, and $n$ additional atoms $\nu_1,\ldots , \nu_n$,
such that there are 6 pipes between $\nu_0$ and each $\nu_i$, for $i=1,\ldots , n$.
In particular, all pipes are incident to the center $\nu_0$.

We put $g(\nu_0)=0$ and $g(\nu_1)=\ldots=g(\nu_n)=1$.
We describe $\varphi$ via a construction of local graphs
$G_0=G_\varphi(\nu_0)$, $G_1=G_\varphi(\nu_1)$, $\ldots$, $G_n=G_\varphi(\nu_n)$.
For ease of presentation, we describe the local graphs as \emph{semi-directed graphs}
(in which some edges are directed and others are undirected). However, in the eventual
instance $(\varphi,g)$, all local graphs are undirected (by replacing every directed
edge with an undirected edge.)

The local graphs $G_1,\ldots, G_n$ are each isomorphic to $K_{3,4}$.
Let $G_0$ be a disjoint union of the semi-directed graphs $G_0^a=C^a\square P_1$, for all $a\in \mathcal{A}$, where $P_1$ is a path of length 1 and $C^a$ is a directed cycle whose length is equal to four times the number of occurrences of $a$ in the clauses in $\mathcal{C}$.
Let $C_P^a$ and $C_N^a$ denote the two vertex disjoint directed induced cycles in $G_0^a$ of length equal to the length of $C^a$,
whose orientation is inherited from $C^a$; see Fig.~\ref{fig:var} for an illustration.

Next, we define the pipes in $E(H)$ by designating the pairs of corresponding virtual vertices in the local graphs;
furthermore, for if a pipe $\rho\in E(H)$ corresponds to virtual vertices $u$ and $v$ in two local graphs,
we also specify a bijection between the set of edges incident to $u$ and the set of edges incident to $v$.
All vertices in local graphs that are not designated to be virtual will be ordinary.
This uniquely determines the instance $(\varphi,g)$.

For each clause $C_i\in \mathcal{C}$, $i\in [n]$, we define three pipes in $E(H)$. Assume that $C_i=(\ell_1 \lor \ell_2 \lor \ell_3)$.
Recall that $G_i$ is isomorphic to $K_{3,4}$. Label the three vertices of degree 4 in $G_i$ by $-1$, $0$, and $1$ resp.; two arbitrary cubic vertices by the literal $\ell_1$; and the remaining two cubic vertices by the literals $\ell_2$ and $\ell_3$, respectively. Let every cubic vertex $v\in V(G_i)$ be virtual, and let each vertex with label $\ell\in \{\ell_1,\ell_2,\ell_3\}$ correspond a vertex $u$ in $C_{P}^a$ if $\ell=a$ and a vertex $u$ in $C_{N}^a$ if $\ell=\neg{a}$. We construct the bijection between the edges incident to $v$ and the edges incident to $u$ as follows: Let the edge between $v$ and the vertex of $G_i$ with label $-1$, $0$, and $1$, resp., correspond to the incoming, undirected, and outgoing edge incident to $v'$ in $G_0$. This completes the definition of the instance $(\varphi,g)$

It remains to prove that $(\varphi,g)$ is a positive instance if and only if $(\mathcal{A}, \mathcal{C})$ is a positive instance.
Assume that $(\varphi,g)$ is a positive instance of \textsc{Generalized atomic embeddibility}. Let $\mathcal{E}:G\rightarrow \mathcal{H}$ be a generalized atomic embedding of $G$ with respect to $(\varphi,g)$. Since $g(\nu_0)=0$, the restriction of $\mathcal{E}$ on $\mathcal{S}(\nu_0)$ yields an embedding of $G_0$ in the plane;
and an embedding of $G_i$ in the torus for all $i\in [n]$.
We construct a satisfying assignment $\tau:\mathcal{A} \rightarrow \{\texttt{true},\texttt{false}\}$ based on the embedding of $G_0$ as follows. We put $\tau(a)=\texttt{true}$ if the incoming, undirected, and outgoing edges incident to a vertex $v$ in $C_P^a$ appear in this counterclockwise order in the rotation of $v$ in the embedding of $G_0$; and we put $\tau(a)=\texttt{false}$ otherwise. Note that the truth value of $a$ is independent of the choice of $v$.
Also note that a literal $\ell$ of the clause $C_i$, $i\in [n]$, is satisfied if and only if the edges between the vertex $v$ labeled by $\ell$ and the vertices labeled by $-1$, $0$ and $1$ in $G_i$ appear in this clockwise order in the rotation at $v$ in the embedding of $G_i$.
Note that Observation~\ref{obs:atomic} holds also for generalized atomic embeddability. By Lemma~\ref{lem:k34} and Observation~\ref{obs:atomic}, every clause $C_i$ must be satisfied by at least 1 and at most 2 literals, and hence, $(\mathcal{A},\mathcal{C})$ is positive.

Now assume that $(\mathcal{A},\mathcal{C})$ is a positive instance of \textsc{Not-All-Equal 3SAT}. We can easily reverse the argument in the previous paragraph as follows. Let $\tau$ be a satisfying assignment witnessing that the instance is positive.  We define a toroidal embedding of $G_i$, for all $i\in [n]$, as follows. For every literal $\ell$, the edges between a vertex $v$ labeled by $\ell$ and the vertices labeled by $-1$, $0$, and $1$ in $G_i$ appear in this clockwise order in the rotation of $v$ in the embedding of $G_i$ if and only if $\ell$ is satisfied by $\tau$. Lemma~\ref{lem:k34} and Observation~\ref{obs:atomic}  imply that we can construct the desired embeddings of $G_1,\ldots , G_n$ on tori.
This also determines a desired spherical embedding of $G_0$, which concludes the proof.
\end{proof}

\section{Thickenability and Connected \textsc{SEFE-2}}
\label{sec:c-planarity}

In this section, we give a polynomial-time reduction of \textsc{atomic embeddability} (Problem~\ref{prob:atomic}), as well as \textsc{Connected sefe-2} (defined below) to \textsc{thickenability} (Problem~\ref{prob:thickenability}).
Recall that an instance of \textsc{atomic embeddability} is given by a simplicial map $\varphi:G \rightarrow H$. The instance is \textbf{positive} if its output answer is \texttt{True}.

\paragraph{Thickenability.}

In the following, we express \textsc{thickenability} as a combinatorial problem, of which \textsc{atomic embaddability} is a proper generalization, and then reduce \textsc{atomic embaddability} to this problem.

Let $P=(H,F)$ denote a 2-polyhedron, where the multigraph $H$ is the 1-skeleton of $P$, and $F$ is the set of facets in $P$, each represented by a cycle in $H$. Let $F=\{f_1,\ldots, f_{|F|}\}$.
A 2-polyhedron $P$ \textbf{embeds} in an orientable 3-manifold $M$ (such as $\mathbb{R}^3$) if the following holds. The multigraph $H$ embeds in $M$ so that the facets $f_1,\ldots, f_{|T|}$ are mapped into pairwise interior disjoint topological discs $D_1,\ldots, D_{|T|}$, resp., in $M$  such that for every $i\in [|F|]$ the boundary of $D_i$, denoted by $\partial D_i$, consists of the  embedded cycle $f_i$. The representation of $P$ in $M$ given by the discs $D_1,\ldots, D_{|F|}$ is an \textbf{embedding} of $P$. The restriction of the embedding of $P$ to the boundaries of these discs gives the  embedding of $H$.

For every $v\in V(H)$, the \textbf{link} of $v$ in $P$ is a multigraph $L_P(v)=(E(v),F(v))$, where the vertex set $E(v)$ is the multiset of edges in $H$ incident to $v$, and the multiset of edges $F(v)$ is in a bijection with the set of facets in $P$ that are incident to $v$ and we give it next. Every pair $\{e,g\}\subset E(v)$ corresponds to an edge in $F(v)$ whose multiplicity equals the number of facets in $F$ that contain both $e$ and $g$.

If $P=(H,F)$ is thickenable, then the intersection of its embedding in a manifold with a sufficiently small 2-sphere centered at (the embedding of) a vertex $v\in V(H)$ is a spherical embedding of the link $L_P(v)$. Indeed, the 2-sphere intersects edges of $H$ incident to $v$ in points, and it intersects facets in $F$ incident to $v$ in Jordan arcs between these points.

Given a polyhedron $P=(H,F)$, the family $\{\mathcal{E}_v : v\in V(H)\}$, where $\mathcal{E}_v$ is a planar (spherical) embedding of $L_P(v)$, is \textbf{compatible} if, for every $e\in E(H)$ joining vertices $u$ and $v$, the rotation at $e$ in $\mathcal{E}_u$ is opposite to the rotation of $e$ in $\mathcal{E}_v$. The observation in the previous paragraph proves the ``only if'' part of the following theorem.

\begin{theorem}[Neuwirth~\cite{N68_thick}] \label{thm:Neuwirth}
The 2-dimensional polyhedron $P=(H,F)$ is thickenable if and only if there exists a family of compatible embeddings of the vertex links of $H$.
\end{theorem}

 We show that testing the condition of Theorem~\ref{thm:Neuwirth} generalizes \textsc{atomic embeddability} and \textsc{connected sefe-2}.

\paragraph{Reduction.}
For a given instance $\varphi:G\rightarrow H$ of \textsc{atomic embeddibility}, we define a 2-dimensional polyhedron $P(\varphi)$, and then show
that $P(\phi)$ is thickenable if and only if $\varphi$ is a positive instance of \textsc{atomic embeddibility}.
Let $P(\varphi)=(\widehat{H},F)$ denote the 2-dimensional polyhedron, where $\widehat{H}$ and $F$ are defined as follows.

Roughly speaking, $\widehat{H}$ is obtained by doubling the multigraph $H$ and connecting each pair of corresponding vertices in the two copies of $H$ by new multiple edges, where every such edge corresponds to a vertex of $G$.
Hence, the vertex set $V(\widehat{H})$ of $\widehat{H}$ is $\{(\mu,0),(\mu,1): \mu\in V(H)\}$, and we define its edge set as

\vspace{-\baselineskip}
\begin{align*}
E(\widehat{H})=& \{\rho_u=(\mu,0)(\mu,1) : \mu=\varphi(u) \mbox{ \rm for some } u \in V(G)\} \cup\\
\cup & \{(\rho,0)=(\mu,0)(\nu,0), (\rho,1)= (\mu,1)(\nu,1): \ \rho=\nu\mu \in E(H) \}.
\end{align*}
The facets in $F$ are in bijection with the edges in $E(G)$. Formally,

\vspace{-\baselineskip}
\begin{align*}
F=& \{(\rho_u=( \mu,0)(\mu,1), \rho_v=(\mu,1)(\mu,0)) : uv\in E(G) \mbox{ \rm such  that } \varphi(u)=\mu=\varphi(v) \} \cup\\
\cup & \{((\rho,0),\rho_u, (\rho,1), \rho_v):  uv\in E(G) \mbox{ \rm such that } \varphi(u)=\mu, \ \varphi(v)=\nu, \mu\not=\nu, \ \mathrm{and} \ \rho=\varphi(uv) \}.
\end{align*}
Thus, for every edge of $uv\in E(G)$, if $\varphi$ maps the vertices $u$ and $v$ to the same atom (resp., different atoms) of $H$,
then the edge $uv$ corresponds to a facet of $P$ bounded by 2 (resp., 4) edges.

It remains to prove that the polyhedron $P(\varphi)$ has the desired property.

\begin{lemma}
\label{lem:main}
For every instance $\varphi:G\rightarrow H$ of \textsc{atomic embeddability},
the 2-dimensional polyhedron $P(\varphi)$ is thickenable if and only if $\varphi$ is a positive instance.
\end{lemma}

\begin{proof}
In order to prove the ``only if'' part, we assume that the polyhedron $P=P(\varphi)=(\widehat{H},F)$ is thickenable. By Theorem~\ref{thm:Neuwirth}, there exists a family of compatible spherical embeddings $\{\mathcal{E}_v : v\in V(\widehat{H})\}$ of the links $\{L_v(P): v\in V(\widehat{H})\}$. Let $\{S_v: v\in V(\widehat{H})\}$ be a family of pairwise disjoint 2-spheres,
and let $\mathcal{E}_v: L_v(P)\rightarrow S_v$ be compatible embeddings for all $v\in V(\widehat{H})$.
We construct an atomic embedding of $G$ on $\mathcal{H}$ with respect to $\varphi$ in two steps as follows.

First, for every $(\mu,0)\in V(\widehat{H})$, we drill a small hole on $S_{(\mu,0)}$ around (the embeedding of) every vertex of $L_{(\mu,0)}$ of the form $(\mu,0)(\nu,0)$. Note that for every such vertex of $V(\widehat{H})$, we have $\mu\nu\in E(H)$ by construction. Let $f_1,\ldots, f_{\deg((\rho,0))}$ denote the edges that are incident to $(\rho,0)=(\mu,0)(\nu,0)$ in $L_{(\mu,0)}(P)$, for some $\nu\mu=\rho\in E(H)$. Formally, we remove from the 2-sphere $S_{(\mu,0)}$ a small open disc $D_{(\mu,0)}((\rho,0))$,
thereby shortening edges $f_1,\ldots, f_{\deg((\rho,0))}$ incident to $(\rho,0)$ in $\mathcal{E}_{(\mu,0)}$ into Jordan arcs $a_{(\rho,0)}(f_1),\ldots,a_{(\rho,0)}(f_{\deg((\rho,0))})$, resp., ending on  $\partial \overline{D}_{(\mu,0)}((\rho,0))$,
which is the boundary of the closure of ${D}_{(\mu,0)}((\rho,0))$.
Let $S_{(\mu,0)}'$ denote the resulting 2-sphere with the removed disc(s) for all $(\mu,0)\in V(\widehat{H})$.

Second, for every edge $\rho=\mu\nu$, we identify the curves $\partial \overline{D}_{(\mu,0)}((\rho,0))$ with $\partial \overline{D}_{(\nu,0)}((\rho,0))$ via a homeomorphism that identifies the endpoints of $a_{(\rho,0)}(f)$ and $a_{(\rho,0))}(f)$ for every $f\ni (\rho,0)$. A desired homeomorphism exists since the embeddings $\mathcal{E}_{(\mu,0)}$ are compatible. Let $\mathcal{{H}}$ denote the surface obtained by the previous identifications, where each $S_{(\mu,0)}$ is interpreted at $S_\mu$. Note that we have just constructed a desired atomic embedding. Indeed, a vertex $u\in V(G)$ is embedded by  $\mathcal{E}_{(\varphi(u),0)}$ as $\rho_u$ on $S_{(\varphi(u),0)}'$;
an edge $uv\in E(G)$ where $\varphi(u)=\varphi(v)$ is embedded by $\mathcal{E}_{(\varphi(v),0)}$ as $\rho_u\rho_v$; and
an edge $uv\in E(G)$ where $\varphi(u)\neq \varphi(v)$ is embedded as the union of $a_{(\varphi(uv),0)}(f)$ on $S_{(\mu,0)}'$ and on $S_{(\nu,0)}'$, where $\varphi(uv)=\nu\mu$, and $f$ corresponds to $uv$ as described in the definition of $F$.
This completes the proof of the ``only if'' part.

For the converse, assume that we are given an atomic embedding of $G$ on $\mathcal{H}$. Clearly, the previous construction  can be reversed to construct a subfamily of compatible spherical embeddings $\mathcal{E}_{(\mu,0)}$ of the links $L_{(\mu,0)}(P)$, for all $(\mu,0)\in V(\widehat{H})$. By taking the mirror image of this construction, we obtain $\mathcal{E}_{(\mu,1)}$, for  all $(\mu,1)\in V(\widehat{H})$.
The union of $\mathcal{E}_{(\mu,0)}$ and $\mathcal{E}_{(\mu,1)}$, for all $\mu\in V(H)$, gives the desired family of compatible spherical embeddings. By Theorem~\ref{thm:Neuwirth}, $P(\varphi)$ is thickenable.
\end{proof}

An immediate consequence of Lemma~\ref{lem:main} is the main result of this section.

\begin{theorem}
\label{thm:main}
\textsc{Atomic Embeddability} reduces to \textsc{thickenability} in polynomial time.
\end{theorem}

\paragraph{Simultaneous embedding with fixed edges (\textsc{SEFE-2}).}
In the following we discuss an implication of Theorem~\ref{thm:main} to the problem of simultaneous embeddability of two graphs \textsc{sefe-2}, which is formally described as follows.

\begin{problem}\label{prob:sefe2}
\textsc{sefe-2.}
Given two (planar) graphs,  $G_1=(V,E_1)$ and $G_2=(V,E_2)$,
decide whether there exists a planar embedding $\mathcal{E}$ of $G=G_1 \cup G_2$
such that both $\mathcal{E}[G_1]$ and $\mathcal{E}[G_2]$ are embeddings.
\end{problem}

The \textsc{Connected sefe-2} is a special case of \textsc{sefe-2} in which $G_1 \cap G_2$ is connected.
Angelini and Da Lozzo~\cite{AngeliniDL16} showed that  \textsc{connected sefe-2} is polynomial-time equivalent to \textsc{c-planarity}. Since \textsc{c-planarity} is a special case of \textsc{Atomic Embeddability},  together with Theorem~\ref{thm:main}, this immediately implies the following.

\begin{corollary}
\label{cor:sefe-2}
\textsc{Connected sefe-2} reduces in polynomial time to \textsc{thickenability}.
\end{corollary}

\section{Acknowledgment}

The first author would like to thank Arnaud de Mesmay for pointing him to Johannes Carmesin's work.
Thanks are extended to Johannes Carmesin for kindly discussing this work with him, and to Jan Kyn\v{c}l for careful reading of the manuscript.

\bibliographystyle{plainurl}
\bibliography{references}

\end{document}